\documentclass[11pt]{article}

\usepackage{amsmath,amsthm,amssymb,amsfonts,cancel}
\usepackage{thmtools,enumitem,subfigure}
\usepackage{mathtools}
\usepackage{dsfont}
\usepackage{tcolorbox}
\usepackage{booktabs}
\usepackage{bbm}
\usepackage{todonotes}
\usepackage{xspace}
\usepackage{tensor}
\usepackage{caption}
\usepackage{bm}
\usepackage{placeins}
\usepackage{hyperref}
\usepackage[capitalize]{cleveref}

\newcommand{\D}{\mathcal{F}}

\newcommand{\Ind}[1]{\mathds{1}_{\left[ #1 \right]}}

\newcommand{\Exp}[1]{\mathbb E \left[ #1 \right]} 

\renewcommand{\Pr}{\mathbb{P}}

\newcommand{\F}{\mathcal{F}}

\newcommand{\HopeOnline}{\textsc{Hope-Online}\xspace}
\newcommand{\HopeFull}{\textsc{Hope-Full}\xspace}
\newcommand{\EtOnline}{\textsc{ET-Online}\xspace}
\newcommand{\EtFull}{\textsc{ET-Full}\xspace}
\newcommand{\MaxMin}{\textsc{MaxMin}\xspace}

\mathchardef\mhyphen="2D 

\DeclarePairedDelimiter{\norm}{\lVert}{\rVert}

\let\originalleft\left
\let\originalright\right
\renewcommand{\left}{\mathopen{}\mathclose\bgroup\originalleft}
\renewcommand{\right}{\aftergroup\egroup\originalright}

\newtheorem{theorem}{Theorem}
\numberwithin{theorem}{section}
\newtheorem{definition}[theorem]{Definition}

\newtheorem{lemma}[theorem]{Lemma}

\newtheorem{proposition}[theorem]{Proposition}

\usepackage[suppress]{color-edits}
\addauthor{sb}{violet}
\addauthor{srs}{orange}
\addauthor{cy}{purple}
\addauthor{gj}{blue}
\usepackage{url}            
\usepackage{booktabs}       
\usepackage{amsfonts}       
\usepackage{nicefrac}       
\usepackage{microtype}      
\usepackage[margin=1in]{geometry}
\begin{document}
\title{Sequential Fair Allocation of Limited Resources under Stochastic Demands}

	\author{
	    Sean R. Sinclair \qquad Gauri Jain \qquad Siddhartha Banerjee \qquad Christina Lee Yu \\
	    Operations Research and Information Engineering, Cornell University \\
	    (\texttt{srs429, gj92, sbanerjee, cleeyu})\texttt{@cornell.edu}}
	   \date{}
\maketitle

\begin{abstract}
We consider the problem of dividing limited resources between a set of agents arriving sequentially with unknown (stochastic) utilities. Our goal is to find a fair allocation -- one that is simultaneously Pareto-efficient and envy-free. When all utilities are known upfront, the above desiderata are simultaneously achievable (and efficiently computable) for a large class of utility functions.  In a sequential setting, however, no policy can guarantee these desiderata simultaneously for all possible utility realizations. 

A natural online fair allocation objective is to minimize the deviation of each agent's final allocation from their fair allocation in hindsight. This translates into simultaneous guarantees for both Pareto-efficiency and envy-freeness. However, the resulting dynamic program has state-space which is exponential in the number of agents.  We propose a simple policy, {\HopeOnline}, that instead aims to `match' the ex-post fair allocation vector using the current available resources and `predicted' histogram of future utilities. We demonstrate the effectiveness of our policy compared to other heurstics on a dataset inspired by mobile food-bank allocations.~\footnote{The code for the experiments is available at \url{https://github.com/seanrsinclair/Online-Resource-Allocation}.}

\end{abstract}
    \newpage
	\setcounter{tocdepth}{2}
	\tableofcontents
	\newpage
	
    \section{Introduction}
\label{sec:introduction}

Our work here is motivated by a problem faced by a collaborating food-bank (Food Bank for the Southern Tier of New York (FBST)~\cite{fbst}) in operating their mobile food pantry program. Every day, the FBST uses a truck to deliver food supplies directly to distribution sites (soup kitchens/pantries/individuals/etc.). When the truck arrives at a site, the operator observes the demand there and chooses how much to allocate before moving to the next site. The number of people assembling at each site changes from day to day, and the operator typically does not know the demand of later sites (but has a sense of the demand distribution based on previous visits). Finally, the amount of food in the truck is usually insufficient to meet the total demand, and so the operator must under-allocate at each site, while trying to be \emph{fair} across all sites. The question we ask is: 

\begin{center}\emph{What is a fair allocation here, and how can it be computed}?\end{center}

In \emph{offline} problems where demands (more generally, utility functions or agent types) for all agents are known to the principal, there are many well-studied notions of fair allocation of limited resources. 
A relevant notion in our context is that a fair allocation is one satisfying three desiderata: \textit{Pareto-efficiency} (for any agent to benefit, another must be hurt), \textit{envy-freeness} (no agent prefers an allocation received by another), and \textit{proportionality} (each agent prefers the allocation received versus equal allocation).
This definition draws its importance from the fact that in many allocation settings it is known to be achievable.
In particular, when goods are divisible, then for a large class of utility functions, an allocation satisfying both is easily computed (via a convex program) by maximizing the Nash Social Welfare (NSW) objective subject to allocation constraints~\cite{varian1973equity,eisenberg1961aggregation}. 

Many practical settings, however, operate more akin to the FBST mobile food pantry, in that the principal makes allocation decisions \emph{online}, with incomplete knowledge of demands (more generally, utility functions) of future agents.  However, these principals do have access to historical data allowing them to generate histograms over utility functions for each agent.  Designing good allocation algorithms in such settings necessitates harnessing the (Bayesian) information of future demands to ensure equitable access to the resource, while also adapting to the online realization of demands as they unfold. 
Guaranteeing Pareto-efficiency, envy-freeness, and proportionality simultaneously turns out to be impossible in such settings (cf.~\cref{lemma:lower_bound}); the challenge thus is in defining meaningful notions of \emph{approximately-fair online allocations}, and developing algorithms which utilize distributional knowledge to achieve such allocations.

\subsection{Motivating Examples}
\label{section: examples}

\noindent \textbf{Mobile Food Pantry.}
Recent demands for food assistance have climbed at an enormous rate, and an estimated fourteen million children are not getting enough food due to the COVID epidemic in the United States~\cite{kulish_2020,brookings}.  With limitations on operating in-person stores, many foodbanks have increased their mobile food bank services. 
In these systems, the mobile food-bank must decide on how much food to allocate to a distribution center on arrival, without knowledge on the demands for locations to come. 

\medskip
\noindent \textbf{Stockpile Allocation.}  In many healthcare systems, states decide how to assign critical resources to patients~\cite{donahue2020fairness,huesch2012one}; for example, the US federal government has recently been tasked with distributing Remdesivir, an antiviral drug used for COVID-19 treatment~\cite{lupkin_2020}.  Another example is assigning patients to psychiatric beds, which have become more and more scarce in recent times~\cite{pinals_fuller_2017}. These decisions are made online, and aim to satisfy individual demands while efficiently using available resources.

\medskip
\noindent \textbf{Reservation Mechanisms.}
These are key for operating shared high-performance computing (HPC) systems~\cite{ghodsi2011dominant}. Cluster centers for HPC receive numerous requests online with varying demands for CPUs and GPUs. Algorithms must allocate resources to incoming jobs, with only distributional knowledge of future resource demands. Important to these settings is the large number of resources (number of GPUs, RAM, etc available at the center), requiring algorithms that scale to higher-dimensional problems.

\subsection{Overview of our Contributions}
\label{section:contributions}

We first formalize the online stochastic fair allocation problem described above, and demonstrate that in the online setting, there are distributions for which no policy can achieve Pareto-efficiency, envy-freeness, and proportionality over all realizations. This motivates studying approximate notions.

For any allocation, a natural \emph{(un)fairness score} is the maximum (alternately, weighted sum) of the deviation of the realized utilities in terms of envy-freeness, (normalized) Pareto-efficiency, and proportionality. 
In the online setting, any such score gives rise to a natural policy to minimize the expected value of this score, which can be formulated as a Markov decision process (MDP). However, since these metrics depend on the entire allocation vector, the complexity of finding the optimal policy is exponential in the number of agents, and also, is difficult to interpret in practice.~\footnote{As an example, consider optimal MDP-based policies for online max-min allocation~\cite{lien2014sequential}.}

Our main conceptual contribution is an alternate objective for online fair allocation, wherein we aim to minimize $\Exp{\norm{X^{opt} - X^{alg}}_{max}}$, the maximum difference between the allocation $X^{alg}$ made by any algorithm, and the \emph{offline} (i.e., ex post) fair solution $X^{opt}$. 
An $\epsilon$ approximation for this objective gives a $c\epsilon$-approximation for the fairness scores defined above (for some problem-specific constant $c$; see~\cref{lemma:e-close}).
The usefulness of this reformulation, though, is not immediately apparent, as it is still a high-dimensional objective. However, we then show how this reformulation allows us to harness recent ideas in model-predictive control to come up with simple algorithms with strong empirical performance.

Our proposed online allocation policy, \HopeOnline, is based on re-solving \textit{information-relaxed} optimization problems, where all future randomness is replaced with expected histograms. 
\HopeOnline is simple, scales to multiple resources, and in experiments, generates allocations close to the optimal fair allocation in hindsight. Moreover, it is balanced across agents, in that the per-agent difference in allocations between earlier and later arriving agents is uniformly small.  Thus, we believe \HopeOnline is a promising candidate for practical online fair allocation.
\srsedit{We do not believe that our work gives the final answer in defining fairness in sequential settings, but hope it will start conversation on how to formally incorporate ethics in sequential AI algorithms.}
    \subsection{Related Work}
\label{sec:related_work_short}

Before proceeding, we discuss some closely related work -- a more extensive survey is provided in~\cref{sec:related_work}.

Fairness in resource allocation, and the use of Nash Social Welfare, was pioneered in seminal work by Varian \cite{varian1973equity,varian1976two}. 
Since then, researchers have investigated fairness properties for both offline and online allocation, in settings with divisible and/or indivisible resources, and when either the agents or resources arrive online; for a comprehensive survey, see~\cite{aleksandrov2019online}. 
Our work focuses on online multi-resource allocation in Bayesian settings; in this context, previous work is mostly limited to non-adaptive algorithms, or consider adversarial arrivals. More importantly, we target additive approximations for individual metrics, instead of approximating global objectives (Eg. maxmin/Nash social welfare) which do not directly give any individual guarantees (see~\cref{app:ce}).

The most common line of work in online fair allocation considers settings where agents are static and items arrive over time~\cite{caragiannis2019unreasonable,zeng2019fairness,ijcai2019-49,bateni2018fair,azar2010allocate}; under stochastic arrivals, these tend to be easier as, intuitively, future allocations can be used to correct past imbalances. 
Closer to our setting are work on online cake cutting; this though is primarily under adversarial arrivals~\cite{walsh2011online,imdynamic,vardidynamic}. 
Finally, recent work considers upfront allocation of indivisible resources for stochastic demands~\cite{donahue2020fairness,elzayn2019fair}; these study similar tradeoffs between global objectives and individual guarantees as us, but are essentially static problems.

In terms of modeling, the closest work to ours is that of \cite{lien2014sequential}, who consider sequential allocation with stochastic demands (arising from similar practical problems with foodbank operations), and propose heuristics for maximizing the minimum utility. Their policies are defined only for single resource settings~\footnote{Note though that in this setting, maxmin allocation is envy-free; this however is not true even for two resources.}.
Using maxmin utility as an objective however leads to some instabilities in allocation policies (for example, a high demand upfront may lead to all future agents getting very small allocations); we demonstrate how our approach improves on this in our experiments.

    \section{Model}
\label{sec:preliminary}

A principal is tasked with dividing $K$ divisible resources among $n$ agents. Each resource $k \in [K]$ has a fixed budget $B_k$ that the principal can allocate.  Each agent $i\in[n]$ has an endowment (or size) $S_i \in \mathbb{R}$ and utility function $u(X_i, \theta_i)$ where $\theta_i \in \Theta$ is a latent type or preference of agent $i$, and $X_i \in \mathbb{R}^{K}$ denotes the \sbedit{normalized} allocation of resources received by agent $i$ (\sbedit{i.e., overall agent $i$ receives $S_iX_i$ units}~\footnote{\sbedit{The endowments $\{S_i\}$ correspond to pre-agreed (deterministic) weights for each agent that reflect their relative sizes; for example a typical foodbank may have $S_i=1$, while another which is twice its size may have $S_{i'}=2$; by normalizing allocations, we can compare them across agents on the same scale.}}). 
We assume the set of types $\Theta$ is finite, and the utility functions $u(X, \theta)$ are $L$-Lipschitz, concave, and strictly increasing with respect to the allocation $X$. Finally, we use $S = \sum_{i=1}^n S_i$ to denote the `effective size' of the population.

In the ex-post or \emph{offline} setting, agents' types $\{\theta_i\}_{i \in [n]}$ are known in advance and can be used by the principal to choose allocations $X \in \mathbb{R}^{n \times K}$ for each agent. 
In the \emph{online} setting the principal visits each agent sequentially in a fixed order $i = 1, \ldots, n$.  
Upon visiting agent $i$, the principal learns their latent type $\theta_i$ \srsedit{drawn from a known distribution $\F_i$}, and must choose allocation $X_i \in \mathbb{R}^K$ before continuing to the next agent. 
Allocation decisions are irreversible, and must obey the overall budget constraints.

\smallskip

\noindent \textbf{Notation}: We use $\mathbb{R}_+$ to denote the set of non-negative reals, and $\norm{X}_{max} = \max_{i,k} |X_{i,k}|$ to denote the matrix maximum norm, and $cX$ to denote entry-wise multiplication for a constant $c$. For allocation $X \in \mathbb{R}_+^{n \times K}$, we use $X_i = (X_{i,1}, \ldots, X_{i,K})$ to denote agent $i$'s normalized allocation, and $B = (B_1, \ldots, B_K)$ the budget vector. 
When comparing vectors, we use $X \leq Y$ to denote that each component $X_i \leq Y_i$. Hence, budget constraints can be written as $\sum_{i=1}^n S_i X_i \leq B$.

\smallskip

\noindent \textbf{Choices of Utility Functions}: In the context of food-bank allocations, we will consider two utility functions of interest.  
With a single resource, a common utility function is the so-called \emph{filling-ratio} $u(X, \theta) = \min\left(\frac{X}{\theta}, 1\right)$~\cite{lien2014sequential}.  
This corresponds to agents having linear utility until their allocation reaches their \emph{demand level} $\theta$, after which the utility is capped.  
While these utility functions are not strictly increasing, we show that all of the results extend to these functions in~\cref{app:proofs}.

For multiple resources, a common choice of utility functions are linear utilities where $u(x, \theta) = \langle \theta, x \rangle$.  Now the latent agent type $\theta \in \mathbb{R}_+^K$ denotes a vector of \emph{preferences} over each of the different resources. More details on modeling agent utilities and sizes are in~\cref{sec:experiments,app:experiments}.

\smallskip

\noindent \textbf{Limitations and Extensions}: The assumption that latent types $\Theta$ are finite is common in decision-making settings, as in practice, this distribution over types is approximated from historical data. 
One limiting assumption is that in the online setting, the principal only knows the latent type of one agent at a time.  In reality the principal could have some additional information about future types (via calling ahead, etc) that could be incorporated in deciding an allocation.  Our algorithmic approach naturally incorporates such additional information.

\subsection{Fairness and Efficiency in Offline Allocations}

\begin{figure}[!t]
	\centering
	\includegraphics[width=.4\columnwidth]{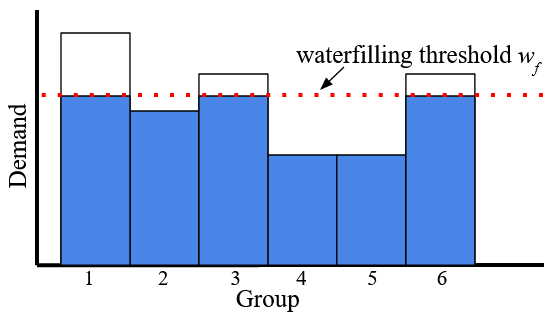}
	\caption{ 
	The waterfilling solution for maximizing NSW with a single divisible resource and filling-ratio utilities (agents on $x$-axis, demands/allocations on $y$-axis). The optimal NSW solution finds a threshold such that the sum of the areas below the demand and threshold equals the budget $B$.
	}
	\label{fig:waterfilling}
\end{figure}

To define a fair allocation in the offline setting (i.e., with known types $\{\theta_i\}_{i \in [n]}$), we adopt an approach proposed by Varian~\cite{varian1973equity}, which is widely used in the OR and economics literature and commonly referred to as `Varian Fairness'.  We will refer to this as fairness for brevity, but for a more detailed discussion on the advantages and limitations of this model, see \cref{app:varian}.
\begin{definition}[Fair Allocation]
\label{def:fairness}
Given types $\{\theta_i\}_{i\in[n]}$, endowments $\{S_i\}_{i \in [n]}$ and utility functions $\{u(\cdot,\theta)\}_{\theta \in \Theta}$, an allocation $X = \{X_{i}\in\mathbb{R}^K_+ \mid \sum_{i=1}^n S_i X_i \leq B\}$ is said to be fair if it simultaneously satisfies the following:
\begin{enumerate}\itemsep0em 
\item \textit{Envy-Freeness} (EF): For every pair of agents $i,j$, we have $u(X_i, \theta_i) \geq u(X_j, \theta_i)$.
\item \textit{Pareto-Efficiency} (PE): For any allocation ${Y}\neq X$ such that $u(Y_i,\theta_i)> u(X_i,\theta_i)$ for some agent $i$, there exists some other agent $j$ such that $u(Y_j,\theta_j)< u(X_j, \theta_j)$. 
\item \textit{Proportional} (Prop): For any agent $i$ we have $u(X_i, \theta_i) \geq u(B/S, \theta_i)$ where $S = \sum_{i=1}^n S_i$.
\end{enumerate}
\end{definition}
\noindent 
While the three properties form natural desiderata for a fair allocation, the power of this definition lies in that asking for them to hold simultaneously rules out many natural (but unfair) allocation policies. 
For example, in the food-bank setting with a single resource and filling-ratio utilities, any allocation that either exhausts the budget or meets the total demand $(\sum_{i} \theta_i)$ is Pareto-efficient. 
One example of this is a \textsc{Greedy} algorithm that assigns $X_i = \theta_i$, until running out of resources.  However, the algorithm is not envy-free as any agent who receives no resources will be envious of an agent who does.
On the other hand, the \textsc{Equal-Allocation} algorithm which assigns $X_i = \frac{B}{S}$ units of resource to each agent trivially achieves EF and PO, but is not necessarily PE (e.g. when $\theta_i < B/S$ for some $i$ and $\theta_j > B/S$ for some $j$). 

More generally, allocation rules based on maximizing a global function such as utilitarian welfare (sum of agent utilities) or egalitarian welfare (the maximin allocation, or more generally, the leximin allocation~\cite{bogomolnaia2001new,lien2014sequential} where one maximizes the minimum utility, and subject to that the second minimum, and so on) are Pareto-efficient, but tend to violate envy-freeness, as they focus on global optimality rather than per-agent guarantees. 
A remarkable exception to this, however, is the Nash Social Welfare $\prod_{i=1}^n u(X_i, \theta_i)^{S_i/S}$:
\begin{proposition}[Theorem 2.3 in \cite{varian1973equity}]
\label{prop:fair}
An allocation that maximizes the Nash Social Welfare is Pareto-efficient, envy-free, and proportional (hence fair).
\end{proposition}
In addition to simultaneously ensuring PE, EF and Prop properties, the NSW maximizing solution can also be efficiently computed via the following convex program called the \emph{Eisenberg-Gale} program~\cite{eisenberg1961aggregation}, obtained by taking the logarithm of the Nash Social Welfare:
\begin{align}
\label{eq:offline_nsw}
\max_{X \in \mathbb{R}_+^{n \times K}} & \,\frac{1}{S} \sum_{i=1}^n S_i  \log\left( u(X_i, \theta_i) \right) 
\\ \text{ s.t. } & \,\sum_{i=1}^n S_i X_i \leq B \nonumber
\end{align}
For a single divisible resource with filling-ratio utilities, the optimal solution to this program is found via a \emph{waterfilling algorithm}, illustrated in~\cref{fig:waterfilling} (see~\cref{sec:experiments} for more details).  Under more general utility functions and multiple resources, the optimal solution is more complex, but can be efficiently computed via standard convex programming techniques~\cite{boyd2004convex}.  
These two properties (that a maximizing allocation for NSW is easy to compute, and satisfies the fairness criteria simultaneously) is key to our proposed online allocation policies.

\subsection{Approximate Fairness in Online Allocations}

Recall that in the online setting the principal visits each agent sequentially in a fixed order $i = 1, \ldots, n$, whereupon visiting agent $i$ the principal sees their latent type $\theta_i \sim \F_i$ and decides on an allocation before continuing to the next agent.  A natural approach to obtain fair allocations in this setting is to develop allocations which satisfy Pareto-efficiency, envy-freeness, and proportionality ex-post.  However, we start with a negative result, showing that such an approach is infeasible even with two agents.
\begin{lemma}
\label{lemma:lower_bound}
For $n = 2$ agents with filling-ratio utilities there exists type distributions $\D_1$ and $\D_2$ such that no online algorithm can guarantee ex-post envy-freeness and Pareto-efficiency almost surely.
\end{lemma}

\begin{proof}
Consider a setting with $B = 2$, and each agent having $S_i=1$ and demand type $\theta \sim \D_i = \text{Uniform}(1 - \delta, 1 + \delta)$ for some arbitrary constant $0 < \delta < 1$.
	
With probability $\frac{1}{2}$ the algorithm will observe that the first agent has a demand of $1 + \delta$.  When the second agent has demand $1 + \delta$ the optimal fair allocation in hindsight will be $X^{opt} = (1,1)$.  When the second agent has demand $1 - \delta$ the optimal solution will be $X^{opt} = (1 + \delta, 1 - \delta)$.  Moreover, each allocation is the unique fair solution for the sequence of demand types $(1 + \delta, 1 +\delta)$ and $(1 + \delta, 1 - \delta)$ respectively.  Hence, no algorithm can achieve the ex-post fair allocation on all sample paths.
\end{proof}

\cref{lemma:lower_bound} shows that simultaneously achieving ex-post envy-freeness and Pareto-efficiency is futile, and hence we need to consider approximate fairness notions. 
While the example is trivial, it highlights the true difficulty in designing online allocations.  Any `fair' online allocation upon visiting agent $i$  must adapt to realized types thus far ($\{\theta_j\}_{j \leq i}$) and exploit the type distribution for future agents $(\{\F_{j})_{j > i}\})$. {It also highlights the difficulty in designing approximate fairness notions, as any such approach must be constructed with individual guarantees in mind.}

On that note, a reasonable modified fairness criteria is to seek allocations $X^{alg} \in \mathbb{R}^{n \times K}$ that minimize the \emph{distance} from (ex-post) envy-freeness, Pareto-efficiency, and proportionality. 
\sbedit{From a practical perspective, moreover, while measuring ex-post envy and distance to proportionality is simple, measuring the distance of an allocation to Pareto-efficiency is not straightforward.
An important proxy, however, is the \emph{resource waste} $B-\sum_{i} S_i X_i$ under any allocation; this follows from observing that any PE algorithm necessarily has no waste (see~\cref{app:proofs} for proof and discussion of an alternative approach)}:
\begin{proposition}
\label{lem:pe-waste}
If an allocation $X \in \mathbb{R}_+^{n \times K}$ is Pareto-efficient, then we have that $\sum_{i} S_i X_i = B$.
\end{proposition}

\noindent We now define our proposed online fairness yardstick:
\begin{definition}[Ex-Post distance from PE, EF, and Prop]
\label{def:distance}
Given agents with types $\{\theta_i\}_{i \in [n]}$ and sizes $\{S_i\}_{i \in [n]}$, and resource budgets $\{B_k\}_{k \in [K]}$, for any online allocation $\{X^{alg}_i\}_{i \in [n]}$, we define:\\
the distance of  $\{X^{alg}_i\}_{i \in [n]}$ to envy-freeness as
\begin{align*}\Delta_{EF} \triangleq \max_{i,j \in [n]} \left(u(X^{alg}_j, \theta_i) - u(X^{alg}_i, \theta_i) \right)\end{align*} 
the distance of  $\{X^{alg}_i\}_{i \in [n]}$ to Pareto-efficiency as 
\begin{align*}\Delta_{PE} \triangleq \max_{k \in [K]} \frac{1}{n}(B_k - \sum_{i} S_i X^{alg}_{i,k})\end{align*} 
the distance of  $\{X^{alg}_i\}_{i \in [n]}$
to proportionality as 
\begin{align*}\Delta_{Prop} \triangleq \max_{i \in [n]} \left( u(B/S, \theta_i) - u_i(X^{alg}_i, \theta_i)\right)\end{align*}
\end{definition}
Note these are all random quantities, depending on both the realized types but also randomness in the allocation algorithm; moreover, any fair allocation necessarily has all of these quantities bounded above by zero. 
The distance to envy-freeness can be thought of as the worst-case envy of any individual, mimicking~\cref{def:fairness}; 
the distance to Pareto-efficiency is taken to be the average per-agent excess of resources wasted by the algorithm; 
the distance to proportionality is the worst-case loss agent experiences under their allocation compared to equal allocations.
The normalization ensures that all are measured on the same scale.

A natural definition of an optimal online fair allocation now is one which minimizes $$\Exp{\max\{\Delta_{PE}, \Delta_{EF}, \Delta_{Prop}\}}$$ (or alternately, any weighted linear combination).  
Given type distributions $\{\D_i\}_{i \in [n]}$, finding such an allocation gives rise to a high-dimensional MDP, as each of the properties depends on the entire matrix of allocations. 
Moreover, since the problem has no obvious structure, the optimal solution may not have a simple form, and can be difficult to interpret.


To get around this, we consider an alternate objective function. Let $X^{opt}$ be the NSW maximizing solution in hindsight (i.e., given realized types $\{\theta_i\}$, solving \cref{eq:offline_nsw}). We instead seek allocations that try to uniformly minimize the expected difference between $X^{alg}$ and $X^{opt}$, which we know to be a fair allocation.
\begin{definition}[$\epsilon$-Fair Allocation]
\label{def:online-fairness}
Given type distributions $\{\F_i\}_{i \in [n]}$ and utility functions $\{u(\cdot, \theta)\}_{\theta \in \Theta}$, we say an online allocation algorithm $X \in \mathbb{R}_+^{n \times K}$ is $\epsilon$-fair if $\sum_{i=1}^n S_i X_i \leq B$ almost surely, and moreover, $$\Exp{\norm{X^{opt} - X^{alg}}_{max}} \leq \epsilon$$
where $X^{opt}$ is the NSW maximizing solution in hindsight.
\end{definition}
One advantage of this definition is that any allocation algorithm which is $\epsilon$-fair also satisfies similar $\epsilon$-additive guarantees in expectation for the earlier defined metrics (\cref{def:distance}) as established below.

\begin{lemma}
\label{lemma:e-close}
Suppose that an algorithm $X^{alg}$ satisfies $\Exp{\norm{X^{opt} - X^{alg}}_{max}} \leq \epsilon$.  Then we have
\begin{itemize}
\item \textit{Approximate Envy-Freeness}:
$$\Exp{\Delta_{EF}} = \Exp{\max_{i,j} \left( u(X_j^{alg}, \theta_i) - u(X_i^{alg}, \theta_i)\right)} \leq {2L \epsilon}.$$
\item \textit{Approximate Pareto-Efficiency}: $$\Exp{\Delta_{PE}} = \frac{1}{n}\Exp{\max_k(B_k - \sum_{i} S_i X_{i,k}^{alg})} \leq \frac{S}{n}\epsilon$$
\item \textit{Approximate Proportionality}: $$\Exp{\Delta_{Prop}} = \Exp{\max_i \left( u(B/S, \theta_i) - u_i(X_i^{alg}, \theta_i)\right)} \leq L \epsilon$$
\end{itemize}
\end{lemma}
\begin{proof}(see \cref{app:proofs})
Follows directly from definitions and Lipschitzness of utility functions.
\end{proof}
    \section{Approximation Algorithms}
\label{sec:approximation_algorithms}

In this section we present \HopeOnline and its counterpart \HopeFull, scalable algorithms that approximate the Nash Social Welfare solution.  These solutions are motivated by approximation algorithms to dynamic programming solutions generated by resolving relaxed versions of the optimization problems \cite{vera2019bayesian}.  Moreover, the allocation rule for the algorithm corresponds to an easily computable policy, is interpretable, and scales well under multiple resources.

Our main algorithm, \textbf{H}istogram \textbf{o}f \textbf{P}reference \textbf{E}stimates \textbf{Online}, or \HopeOnline, arises from the observation that a natural approximation algorithm arises from replacing unknown quantities in \cref{eq:offline_nsw} with their distribution.  Moreover, in any envy-free allocation, two agents $i$ and $j$ with the same type $\theta$ will acquire the same allocation.  Hence, we can rewrite \cref{eq:offline_nsw} as follows:
\begin{align}
\label{eq:offline_nsw_type}
\max_{X \in \mathbb{R}_+^{|\Theta| \times K}} & \frac{1}{S} \sum_{\theta \in \Theta} \left(\sum_{i = 1}^n S_i \Ind{\theta_i = \theta}\right) \log(u(X_\theta, \theta)) \\
\text{s.t. } & \sum_{\theta \in \Theta} \left(\sum_{i=1}^n S_i \Ind{\theta_i = \theta}\right) X_\theta \leq B
\nonumber
\end{align}
where the allocation to any agent $i$ with type $\theta$ is $X_\theta$.
The \HopeOnline algorithm now approximates the NSW allocation in \cref{eq:offline_nsw_type} by re-solving the above program while replacing the unknown quantities $\Ind{\theta_i = \theta}$ with their expectation $(\Pr(\theta_i = \theta))$.
In more detail, for the $i^{th}$ agent, given the current budget vector $B^i$, we re-solve the Eisenberg-Gale program in~\cref{eq:offline_nsw_type} with all future demand replaced by the expected future histogram over types. Formally, at iteration $i$, the algorithm observes the latent type $\theta_i$ for agent $i$ and allocates $X_i^{alg} = X_{\theta_i}$ according to the solution to:
\begin{align*}
\max_{X \in \mathbb{R}_+^{|\Theta| \times K}} & \frac{1}{S} \sum_{\theta \in \Theta} N_\theta \log(u(X_\theta, \theta)) \\
\text{s.t. } & \sum_{\theta \in \Theta} N_\theta X_\theta \leq B^i.
\end{align*}
where $B^i = B^{i-1} - X_{i-1}^{alg}$ is the current available resources, and the expected histogram over types is defined as 
\begin{align*}
N_\theta = S_i \Ind{\theta_i = \theta} + \sum_{j=i+1}^n S_j \Pr(\theta_j = \theta)
\end{align*} 
Note here that as the type $\theta_i$ for agent $i$ is observed, the probability for agent $i$ is replaced by the Dirac-$\delta$ function on the observed value.

An alternative algorithm is \textbf{H}istogram \textbf{o}f \textbf{P}reference \textbf{E}stimates \textbf{Full}, or \HopeFull, which follows the same idea but instead solves the Eisenberg-Gale program with all agents (including agents already visited) to decide an allocation.  At iteration $i$, the algorithm observes the types $\{\theta_j\}_{j \leq i}$ and allocates to agent $i$ the allocation $X_i^{alg} = \min(X_{\theta_i}, B^i)$ according to the solution to:
\begin{align*}
\label{eq:full_weighted}
\max_{X \in \mathbb{R}_+^{|\Theta| \times K}} & \frac{1}{S} \sum_{\theta \in \Theta} \tilde{N}_\theta \log(u(X_\theta, \theta)) \\
\text{s.t. } & \sum_{\theta \in \Theta} \tilde{N}_\theta X_\theta \leq B, 
\end{align*}
and the expected histogram over types now given by:
\begin{align*}
 \tilde{N}_\theta  = \sum_{j=1}^{i}S_j\Ind{\theta_j = \theta} + \sum_{j=i+1}^n S_j\Pr(\theta_j = \theta).    
\end{align*}
This forms a natural estimator for~\cref{eq:offline_nsw_type}, in that for the last agent, it solves the exact same optimization problem as~\cref{eq:offline_nsw_type}.  However, experiments show that its performance is worse than \HopeOnline, as the algorithm uses the original budget $B$ instead of the remaining budget.

\begin{table*}[!ht]
\caption{
Comparison of fairness metrics (averaged over 1000 replications) on the single-resource online allocation problem with filling-ratio utilities. We compare the four unfairness metrics from~\cref{def:distance,def:online-fairness} (larger values correspond to lower scores; best value highlighted), over the two datasets described in \cref{sec:experiments_single} (i.i.d discretized Gaussian demands with $n = 100$, demands generated from FBST dataset with $n = 6$). \HopeOnline is best or second-best across all metrics and settings (note that \textsc{Greedy} naturally minimizes waste, while \textsc{Adaptive-Threshold} ensures Proportionality by design).
} \label{tab:fairness}
\setlength\tabcolsep{0pt} 
\footnotesize\centering
\begin{tabular*}{\textwidth }{@{\extracolsep{\fill}}r*8c}
\toprule
Algorithm &  \multicolumn{2}{c}{$\Exp{\norm{X^{opt} - X^{alg}}_{max}}$} & \multicolumn{2}{c}{$\Exp{\Delta_{EF}}$} & \multicolumn{2}{c}{$\Exp{\Delta_{PE}}$} & \multicolumn{2}{c}{$\Exp{\Delta_{Prop}}$}\\
\midrule
{}   & $n=6$   & $n=100$ & $n=6$   & $n=100$ & $n=6$   & $n=100$ & $n=6$   & $n=100$\\
\HopeOnline   &  \textbf{1.23} &  \textbf{2.16} &  0.059 & \textbf{0.11} &  0.35 &  0.14 &  0.057 &  0.012 \\
\HopeFull   &  1.37 &  4.68 &  0.076 &  0.26 &  0.38 &  0.15 &  0.076 &  0.021 \\
\EtOnline   &  1.47 &  3.57 &  0.091 & 0.20 &  0.38 &  0.22 &  0.090 &  0.070 \\
\EtFull   &  1.37 &  3.82 &  0.076 &  0.19 &  0.47 &  0.24 &  0.076 &  0.065 \\
\MaxMin   &  1.34 &  2.87 &  0.064 &  0.14 &  0.36 &  2.02 &  0.062 &  0.13 \\
\textsc{Greedy}   &  1.72 &  6.68 &  0.14 &  0.40 &  \textbf{0.28} &  \textbf{0.13} &  0.13 &  0.37 \\
\textsc{Adaptive-Threshold}   &  16.22 &  5.09 &  \textbf{0.00059} &  0.22 &  4.66 &  0.70 &  \textbf{0} &  \textbf{0} \\
\bottomrule
\end{tabular*}
\end{table*}

\section{Experiments}
\label{sec:experiments}
\begin{figure*}[!t]
  \begin{minipage}[c]{0.49\textwidth}
    \centering
    \includegraphics[width=\textwidth]{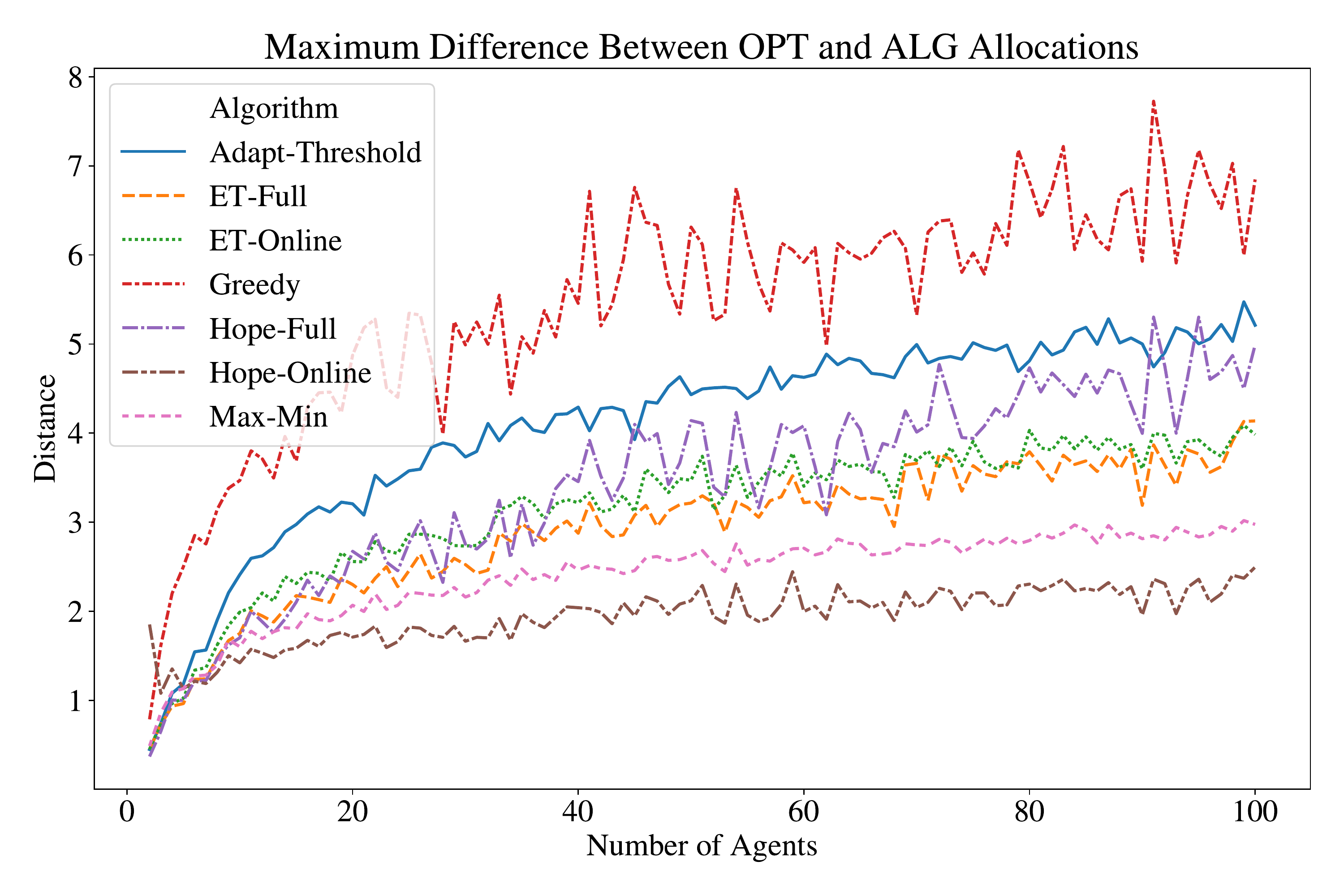}
    	\caption{}
 	\label{fig:regret}
  \end{minipage}
  \hfill
  \begin{minipage}[c]{0.49\textwidth}
    \centering
    \includegraphics[width=\textwidth]{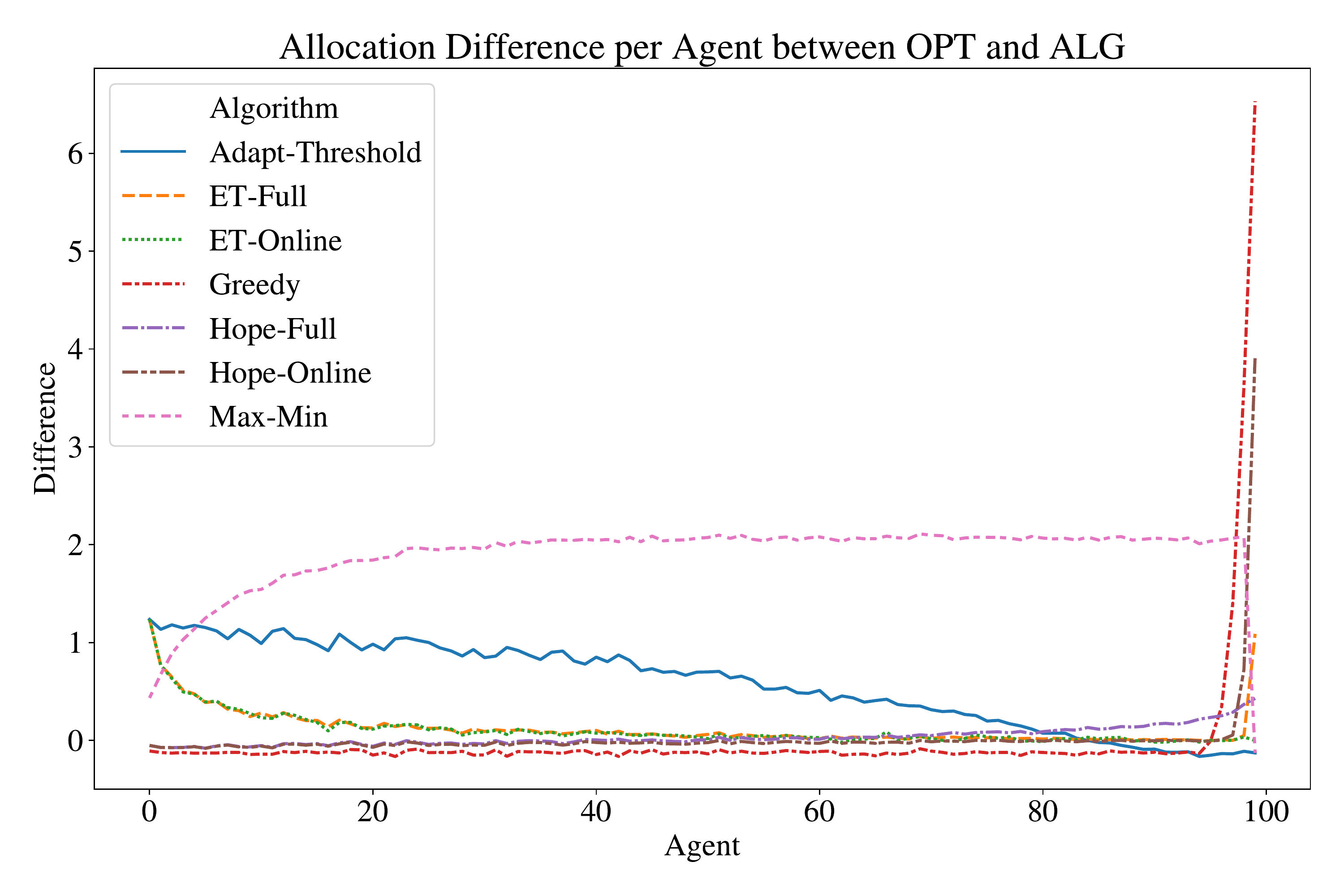}
    	\caption{}
 	\label{fig:allocation}
  \end{minipage}
  \caption*{
  (\cref{fig:regret}) Number of agents on $x$-axis, and $\Exp{\norm{X^{opt} - X^{alg}}_\infty}$ on y-axis;  (\cref{fig:allocation}) The $x$-axis corresponds to a specific agent $i$ in a simulation, and the $y$-axis corresponds to $\Exp{X_i^{opt} - X_i^{alg}}$, the per-agent deviation in the final allocation vs the optimal offline one. }
\end{figure*}

We now test our algorithms on both single and multiple resource allocation settings, with experiment parameters based on food-bank allocation data. 
For the full experimental details, details on the dataset, run-time analysis, and measures of variance of these results see~\cref{app:experiments}. 

In each experiment, we compare the allocations $X^{alg}$ given by the \HopeOnline \, algorithm to the NSW maximizing allocation in hindsight $X^{opt}$ (i.e. the solution to~\cref{eq:offline_nsw_type}). 
We report multiple performance measures, including the maximum allocation deviation $\Exp{\norm{X^{alg} - X^{opt}}_{max}}$, the agent-by-agent difference in allocations $\Exp{X_i^{alg} - X_i^{opt}}$, and expected approximate envy-freeness, Pareto-efficiency, and proportionality (see \cref{def:distance}). Finally, to benchmark the performance of \HopeOnline, we simulate several alternate heuristics, including the \HopeFull heuristic, as well as other optimization-based approaches discussed in \cref{app:full_heuristics}.

\subsection{Single Resource Allocation}
\label{sec:experiments_single}

We first consider a simple single-resource variant of the food-bank allocation problem motivated in~\cref{section: examples}.  A mobile food pantry loads up the truck at the start of the day with a fixed number of `meals' $B$, and travels sequentially from one drop-off location to the next. In each location $i$, they observe a demand $\theta_i \in \mathbb{R}_+$, make an allocation, and then proceed to the next food drop-off point. 
We consider filling-ratio utilities $u(X, \theta) = \min(\frac{X}{\theta}, 1)$.  These utility functions are of interest to food-banks as it serves as a common metric used in evaluating the effectiveness of a food bank.  In particular, the filling ratio a food bank is able to provide to each of the distribution sites often serves as one component of a `score' used in markets deciding how many resources a food bank receives~\cite{prendergast2017food}.

We start by including a structural result.  While the filling-ratio utilities chosen are not monotonically increasing (due to the $\min$), the Eisenberg-Gale program in \cref{eq:offline_nsw} still guarantees a fair allocation.  Moreover, the optimal solution can be characterized via a Waterfilling solution (see \cref{fig:allocation}).
Since \HopeOnline also allocates resources according to the solution to an optimization problem of a similar form (albeit with different weights, based on current available information), the resulting allocation also takes a waterfilling form, with the waterfilling threshold adjusted over time.
\begin{theorem}
\label{thm:single-resource}
An optimal solution to~\cref{eq:offline_nsw} for a fixed set of agent demands $\{\theta_i\}_{i \in [n]}$ is given by a waterfilling policy
$$X_i^{opt} = \min\{w_f, \theta_i, B\},$$ where the waterfilling threshold $w_f$ solves $$\min\left(B, \sum_{i=1}^n S_i \theta_i\right) = \sum_{i=1}^n S_i \theta_i \Ind{\theta_i \leq w_f} + S_i w_f \Ind{\theta_i \geq w_f}.$$
Moreover, this allocation is Pareto-efficient, envy-free, proportional, and hence fair.
\end{theorem}

\begin{table*}[!t]
\caption{ 
Comparison of fairness metrics averaged over 1000 simulations on the multiple-resource online allocation problem with linear utilities.  We plot the four metrics from \cref{def:online-fairness} and \cref{def:distance}.  Larger values corresponds to a lower score on that metric.} \label{tab:mult-fairness}
\setlength\tabcolsep{0pt} 
\footnotesize\centering
\begin{tabular*}{\textwidth }{@{\extracolsep{\fill}}r*4c}
\toprule
Algorithm &  {$\Exp{\norm{X^{opt} - X^{alg}}_{max}}$} & {$\Exp{\Delta_{EF}}$} & {$\Exp{\Delta_{PE}}$} & {$\Exp{\Delta_{Prop}}$}\\
\midrule
\HopeOnline   &  \textbf{0.0090} &  \textbf{0.039} &  \textbf{0.011} & 0.0060 \\
\HopeFull   &  0.019 & 0.16 &  0.024 & 0.0094 \\
\EtOnline   &  0.011 &  0.070 &  0.027 & \textbf{0.0027} \\
\EtFull   &  0.20 &  0.17 &  0.025 & 0.0094 \\
\bottomrule
\end{tabular*}
\end{table*}

\medskip \noindent
\textbf{Choice of Demand Distribution.} 
We perform several synthetic experiments, with parameters based on food bank demand data.
We present results for two simulations here, and defer others to the appendix.  The first simulation considered the performance of \HopeOnline with respect to the number of agents $n$.  
Here we model the demand distribution $\F_i$ as an i.i.d. discretized Gaussian distribution with mean $15$ and variance of $3$.  
The second simulation considers a more realistic scenario, where the number of agents $n = 6$ (corresponding to the six counties in the Southern Tier of New York) and a demand distribution as a discretized Gaussian with mean and variance generated from a demand histogram collected from a dataset on food demands by county.

In the experiments we set endowment variables $S_i = 1$ as the demand $\theta_i$ represents the `size' of an agent.  
We also set the budget as the expected sum of demands.
From a practical perspective, this is a typical choice made by food banks as to how many meals/boxes to prepare; theoretically, this is meaningful, as the optimal solution will give everyone an equal allocation if the budget is smaller than $n$ times the smallest possible demand, and give everyone their demand if the budget is larger than $n$ times the largest possible demand, both of which make the online problem trivial. The expected sum of demands is between these two bounds, forcing the online algorithms to make non-trivial decisions that can significantly impact their performances.

\medskip \noindent
\textbf{Heuristics.} We compare \HopeOnline to several natural online fair-allocation heuristics:
\begin{itemize}[nosep,leftmargin=*]
\smallskip
\item \textsc{Greedy}: every agent is given its demand $\theta_i$ until the budget is depleted.
\smallskip
\item \textsc{Adaptive-Threshold}: Allocate $X_i = \min(B_i / (n-i), \theta_i)$, the equal share of the remaining budget.
\smallskip
\item \MaxMin: The heuristic proposed in \cite{lien2014sequential} for optimizing the maxmin allocation. 
\smallskip
\item \HopeFull: Re-solves~\cref{eq:offline_nsw_type} over \emph{all} agents, using predictive-histogram for future agents (cf.~\cref{sec:approximation_algorithms}).
\smallskip
\item \EtOnline, \EtFull: Alternate model-predictive heuristics based on re-solving~\cref{eq:offline_nsw_type} with \emph{expected utility functions} (See~\cref{app:full_heuristics})
\end{itemize} 



\medskip \noindent
\textbf{Distance to Fair Allocation in Hindsight.} In~\cref{fig:regret} we compare $\Exp{\norm{X^{opt} - X^{alg}}_\infty}$, the difference in the allocation generated by the algorithm and the optimal fair solution in hindsight as we vary the number of agents $n$.  We see that \HopeOnline generates allocations close to the optimal solution in hindsight.  Unsurprisingly, the \textsc{Greedy} and \textsc{Adaptive-Threshold} algorithms perform worse with an increasing number of agents.  Thus \HopeOnline does well at ensuring approximate individual fairness guarantees (\cref{def:online-fairness}).  Moreover, for the $\ell_1$ norm between the allocations, our experiments show \HopeOnline has sub-linear dependence with $n$ (see~\cref{app:experiments}).

\medskip \noindent
\textbf{Group-Based Differences in Allocation.}  In~\cref{fig:allocation}, we study if the algorithms tend to under/over allocate based on the order of arrival. 
For this, we set $n = 100$, and compare $\Exp{X_i^{opt} - X_i^{alg}}$ for each agent $i = 1,\ldots, 100$ (averaged over $100$ replications).  \HopeOnline gives a uniform approximation to the offline NSW solution, with small performance degradation for later agents. 
In contrast, \MaxMin is too pessimistic and under-allocates to all agents.

\medskip \noindent
\textbf{Online Fairness Metrics.} In~\cref{tab:fairness} we compare each algorithm on the fairness metrics for a fixed number of agents (on the two problem set-ups described).  With these metrics we see that \HopeOnline performs competitively across a wide variety of distributions and problem set-ups.  Moreover, in~\cref{app:experiments} we show further simulation results, showing that \HopeOnline even outperforms \MaxMin on the metric of the minimum fill rate $(\Exp{\min_i u(X_i, \theta_i)})$, a metric \MaxMin was designed to maximize!

\subsection{Multiple Resource Allocation}

We next consider multiple resource allocation settings, with linear utility functions $u(x, \theta) = \langle \theta, x \rangle$. 
We use $n = 6$ agents, with sizes $S_i$ corresponding to population of the six counties of the Southern Tier of New York. 
For the utility functions, we use data from~\cite{prendergast2017food} on the `prices' on different resources in the non-monetary mechanism used to allocate resources to major food bank distribution centers; we consider a subsection of products presented: cereal, diapers, pasta, paper, prepared meals, rice, meat, fruit, and produce. 
To generate our preference distributions, we created eight different preference profiles, with each $\theta_k = \text{Bernoulli}(1/2) w_k$ where $w_k$ is the price of the product. The type distribution is chosen to be uniform over these types.

\medskip \noindent \textbf{Heuristics} We again test \HopeOnline against the alternative heuristics introduced before (note that \textsc{Greedy}, \textsc{Adaptive-Threshold} and \MaxMin do not extend to multiple resources).

An additional heuristic here is \textsc{Proportional} allocation, which assigns $X_i = B / S$ to each agent. By design, this allocation satisfies $\Delta_{PE}, \Delta_{EF},$ and $\Delta_{Prop} = 0$.  However, the allocation can be arbitrarily bad in terms of Pareto-efficiency, and so we omit it from the comparison here.

\medskip \noindent \textbf{Online Fairness Metrics.} In \cref{tab:mult-fairness} we compare each algorithm on the fairness metrics for a fixed number of agents $n = 6$ averaged over $100$ simulations.  With these results we again see that \HopeOnline performs competitively across each benchmark compared to the other algorithms.  Moreover, it performs the best in terms of minimizing the distance to the optimal allocation in hindsight, $\Exp{\norm{X^{opt} - X^{alg}}_{max}}$.

    \section{Conclusion}
\label{sec:conclusion}

In this paper we considered online allocation of divisible resources.  In the offline setting achieving a fair allocation scheme is found by maximizing the Nash Social Welfare objective.  In the online setting, however, we showed that no algorithm can achieve ex-post fairness almost surely.  In light of this, we defined a new notion of approximate fairness, where an approximately fair allocation algorithm is one which is close to the fair Nash Social Welfare maximizing solution in hindsight.  This definition is natural, as any approximately fair algorithm satisfies approximate counterparts to Pareto-efficiency, envy-freeness, and proportionality.

While the usefulness of this objective is not immediately apparent, we show that it leads to a simple algorithm \HopeOnline which approximates the offline solution by solving \textit{information-relaxed} versions of the Eisenberg-Gale program, where unknown quantities are replaced with their histogram.  Through experiments we show that \HopeOnline leads to allocations with fairness properties competitive to several benchmarks and prior work.

\srsedit{Although fairness in resource allocation is well-studied in the offline setting, fairness metrics for the sequential setting are very poorly understood.  Our proposed metrics (\cref{def:distance}) give a novel way to extend Varian's definitions to the sequential setting.  While we do not believe our work gives the final answer in defining fairness in sequential settings, we hope it starts a conversation on how to formally incorporate ethics and fairness constraints in sequential allocation problems.}
    
    \section*{Acknowledgements}
    Part of this work was done while Sean Sinclair and Christina Yu were visiting the Simons Institute for the Theory of Computing for the semester on the Theory of Reinforcement Learning. We also gratefully acknowledge funding from the NSF under grants ECCS-1847393, DMS-1839346, CCF-1948256, and CNS-1955997, the ARL under grant W911NF-17-1-0094, and the Cornell Engaged Grant: Applied Mathematics in Action.
    
    \bibliographystyle{plain}
    {\bibliography{references}}
    \appendix
    \newpage

    \section{Related Work (Detailed)}
\label{sec:related_work}

Fairness in resource allocation, and the use of Nash Social Welfare, was pioneered by Varian in his seminal work~\cite{varian1973equity,varian1976two}. Since then, researchers have investigated fairness properties for both offline and online allocation, in settings with divisible or indivisible resources, and when either the agents or resources arrive online. We now briefly discuss some related works; see~\cite{aleksandrov2019online} for a comprehensive survey. What distinguishes our setting from many of the previous works is that we consider the online Bayesian setting with a known distribution. Many previous works are either limited to offline or non-adaptive algorithms, or consider adversarial online arrivals.


\medskip

\noindent \textbf{Food Bank Operations}: There is a growing body of work in the operations research literature addressing logistics and supply chain issues in the area of humanitarian relief and food distribution~\cite{sengul2017modeling,orgut2016achieving}.  The research focuses on designing systems which balance efficiency, effectiveness, and equity.  In \cite{eisenhandler2019humanitarian} they study the logistical challenges of managing vehicles with limited capacity to distribute food and provide routing and scheduling protocols.  In \cite{lien2014sequential} they consider sequential allocation with an alternative objective of maximizing the minimum utility (also called the leximen in the literature \cite{moulin2004fair}).  We instead consider sequential allocation of resources under the Nash Social Welfare objective to obtain equitable allocations~\cite{varian1973equity}.

\medskip

\noindent \textbf{Cake Cutting}: Cake cutting serves as a model for dividing a continuous object (whether that be a cake, advertisement space, land, etc)~\cite{brams1995envy,procaccia2013cake}.  Under this model, prior work considers situations where individuals arrive and depart during the process of dividing a resource, where the utility of an agent is a set-function on the interval of the resource received.  Researchers analyze the offline setting to develop algorithms to allocate the resource with a minimal number of cuts \cite{brams1996fair}, or online under adversarial arrivals~\cite{walsh2011online}.  Our model imposes stochastic assumptions on the utilities for arriving agents and characterizes probabilistic instead of sample-path fairness criteria.

\medskip

\noindent \textbf{Online Resources}: One line of work considers the resource (here to be thought of as the units of food, processing power, etc) are online and the agents are fixed~\cite{benade2018make,aleksandrov2015online,mattei2018fairness,mattei2017mechanisms,aleksandrov2019monotone}.  In \cite{zeng2019fairness} they study the tradeoffs between fairness and efficiency when items arrive adversarially.  Another common criteria is designing algorithms which are \textit{envy-free up to one item}, where researchers design algorithms that can reallocate previously allocated items, but try to minimize these adjustments~\cite{ijcai2019-49,aziz2016control}.

\medskip

\noindent \textbf{Online Agents}: The other setting considers agents as arriving online and the resources as fixed.  In \cite{kalinowski2013social} they consider this setting where the resources are indivisible with the goal of maximizing utilitarian welfare (or the sum of utilities) which provides no guarantees on envy-freeness.  Another approach in \cite{gerding2019fair} considers a scheduling setting where agents arrive and depart online.  Each agent has a fixed and known arrival time, departure time, and demand.  The goal then is to determine a schedule and allocation which is Pareto-efficient and envy-free.  We instead consider a stochastic setting where each agent has a distribution on utilities and seek algorithms which satisfy probabilistic versions of Pareto-efficiency and envy-freeness.

\medskip

\noindent \textbf{Non-adaptive Allocations}:  A separate line of work considers fairness questions for resource allocations in a similar setting where the utilities across groups are drawn from known probability distributions \cite{donahue2020fairness,elzayn2019fair}.  This  line  of  work  investigates  probabilistic versions of fairness, where the goal is to quantify the discrepancy between the objectives of ensuring the expected utilization of the resources is large (ex-ante Pareto-optimal), while the probability of receiving the resource is proportional across groups (ex-ante proportional).  However, they consider algorithms which decide on the entire allocation for each agent upfront before observing the utilities for any individuals rather than adaptive policies.

\medskip

\noindent \textbf{Adaptive Allocations}: In contrast, we consider a model where the principal makes decisions on how much of the resource to allocate after witnessing an agent's type, where all future types are unknown.   Most similar to our work is recent work analyzing a setting where agents arrive over time and do not depart, so that the algorithm can allocate additional resources to agents who arrived in the past \cite{kash2014no}.  We instead consider a stochastic setting where agents arrive and depart in the same step with the goal of characterizing allocations that cannot reallocate to previous agents.  Other papers either seek competitive ratios in terms of the Nash Social Welfare objective~\cite{azar2010allocate,bateni2018fair,gorokh2020fair}, or derive allocation algorithms which perform well in terms of max-min~\cite{lien2014sequential}.
    \section{Discussion on Varian's Fairness}
\label{app:varian}

\subsection{Limitation of Fair Allocations}

Economists, computer scientists, and people in the operations research literature have become increasingly interested in questions of fairness~\cite{sugden1984fairness,varian1973equity,varian1976two}.  One particular concept of fairness which has gained wide circulation due to its ease in compatibility is the so-called notion of `Varian fairness' pioneered by Hal Varian in the 1970s taken in \cref{def:fairness}.  This theory of fairness provides three criteria for judging a given allocation of resources: \textit{envy-freeness}, \textit{Pareto-efficiency}, and \textit{proportionality}, all defined with respect to utilities an agent has for different allocations.  These criteria serve as more of a classification than an optimization perspective, as each of them merely provides a true/false criteria for an allocation to \textit{satisfy} fairness rather than a way an allocation can \textit{approach} fairness.  Numerous other researchers have proposed other definitions of fairness, including $\alpha$-fairness obtained by instead maximizing~\cite{moulin2004fair,arrow2012social}:
$$\sum_{i=1}^n \text{sign}(\alpha)u(X_i, \theta_i)^{\alpha}.$$
In this definition, taking $\alpha = 1$ recovers utilitarian welfare, or maximizing the sum of utilities.  Taking $\alpha \rightarrow - \infty$ also recovers the leximin objective, and $\alpha \rightarrow \infty$ the leximax objective.

One primary critique of `Varian fairness' is that a Varian fair allocation may not exist at all.  Moreover, the implication is that \textit{if} a Varian-fair allocation exists, then it has special merit.  While we specifically consider settings where a `Varian fair' allocation always exists (and is remarkably found as a result of optimizing the Nash Social Welfare objective), it is important to consider some of the several downsides of this model.

\medskip \noindent \textbf{Comparison of Individuals}: Paramount to Varian's definition of proportionality and envy-freeness is that each agent is treated symmetrically.  This ignores systemic factors that inhibit particular individual's access to the resource. 

\medskip \noindent \textbf{Scale Invariance}: The concept of fairness is strictly operational, in the sense that it requires no more information than what is contained in an agent's utility function.  Settings like matching students to local schools via school choice require definitions which measure the `utility of replacement'~\cite{abebe2020roles}.  As an example, a student with preferences (School A, School B, School C) in descending order gets matched to School B.  How can we measure the overall gain to society when the student is instead matched to School A, or School C in comparison to another students list of preferences?  In Varian's definition of fairness, utility functions are only used to exhibit an ordering on preferences, rather than a relative value on different outcomes.

\medskip

We believe the settings considered in \cref{section: examples} are well suited to Varian's model on fairness.  In the example motivated with the Food Bank of the Southern Tier of New York, agents correspond to individual distribution sites, whether that be a soup kitchen, a drop-off location for the mobile food bank, etc.  In these settings, locations have use for all resources with strictly increasing utility with respect to the resource allocated.  This motivates using scale invariant measures, as every agent will be able to use all available food allocated to them.  Considering processor assignment in cloud computing platforms, each individual request coming in should be treated independently and symmetrically.

\subsection{Competitive Ratio or Individual Guarantees}
\label{app:ce}
One approach on obtaining fairness guarantees for an online algorithm could be in the form of a competitive ratio.  These results find allocation algorithms $X^{alg}$ to which you can construct a bound on the competitive ratio for the Nash Social Welfare (or its logarithm):
\begin{align*}
\quad\quad \frac{\prod_{i=1}^n u(X_i^{alg}, \theta_i)}{\prod_{i=1}^n u(X_i^{opt}, \theta_i)} \quad\quad \text{ or }  \frac{\sum_{i=1}^n \log(u(X_i^{alg}, \theta_i))}{\sum_{i=1}^n \log(u(X_i^{opt}, \theta_i))} .
\end{align*}
While theoretically interesting as the typical Nash Social Welfare (or logarithm in the Eisenberg Gale program) are of a different form than typical competitive ratio guarantees done in computer science, these results provide no immediate individual fairness guarantees.  The motivation for the Eisenberg-Gale program arises from the fact that \textit{fairness is a byproduct}.  To some extent, the actual objective value of an allocation is meaningless, and the objective is only taken as it serves as a proxy to obtain fair allocations.  Ensuring a good competitive ratio has no direct guarantees on individual fairness.  In many applications of resource allocation, stakeholders are more interested in obtaining individualized guarantees than global guarantees on social welfare.  This motivated our alternative approach of designing algorithms with individualized guarantees in mind.
    \section{Heuristic Algorithms}
\label{app:full_heuristics}

Here we describe the other heuristic algorithms considered in \cref{sec:experiments}, and compare our approaches to similar approaches developed in \cite{lien2014sequential} for the max-min objective.

\subsection{Expected Type - Online and Full}
\label{app:et-algo}

A second derivation of heurstic algorithms is to notice that in many practical applications, the agent types $\theta \in \Theta$ are in $\mathbb{R}^{K}$, corresponding to `preferences' for each item.  One simple heuristic approach in designing online allocation schemes is to replace $\theta_i$ in the offline Nash Social Welfare optimization problem (\cref{eq:offline_nsw}) with its expectation $\Exp{\theta_i}$.  This leads to two simple heuristic algorithms which use the \textbf{E}xpected \textbf{T}ypes and either the \textbf{Online} or \textbf{Full} information over previous agents.

In particular, the \EtOnline algorithm at every iteration $i$, observes the type $\theta_i$ and resolves the Nash Social Welfare objective with the current available resources and future agents types are replaced with their expectation.  In particular, at every iteration $i$, the algorithm allocates $X_i^{alg} = X_i$ according to the solution to:
\begin{align*}
\max_{X \in \mathbb{R}_+^{n-i \times K}} & \frac{1}{S} \sum_{j=i}^n S_j (\log(u(X_j, \Exp{\theta_j})) \Ind{j > i} + \log(u(X_i, \theta_i))\Ind{j=i}) \\\text{ s.t. } & \,\sum_{j=i}^n S_j X_j \leq B^i
\end{align*}
where $B^i = B^{i-1} - X_{i-1}^{alg}$ is the current available budget taking into account allocations already committed in previous iterations.

A similar approach would be the \EtFull algorithm that mimics \HopeFull by utilizing all of the prior observed types in designing an allocation.  In particular, this heuristic allocates $X_i^{alg} = \min(X_i, B^i)$ where $X_i$ is the solution to:
\begin{align*}
\max_{X \in \mathbb{R}_+^{n \times K}} & \frac{1}{S} \sum_{j=1}^n S_j (\log(u(X_j, \Exp{\theta_j})) \Ind{j > i} + \log(u(X_i, \theta_j))\Ind{j \leq i}) \\\text{ s.t. } & \,\sum_{j=1}^n S_j X_j \leq B
\end{align*}

Similar to \HopeFull, this solves the exact same optimization problem as the offline optimal fair solution \cref{eq:offline_nsw_type} for the last agent.

One downside to both \EtOnline and \EtFull is that the expected type will not necessarily be in the support of the distribution.  This negatively impacts the allocation returned by the algorithm since the offline Nash Social Welfare solution will only take into account observed types which fall into the support of the distributions.  As observed in \cref{sec:experiments} these algorithms have worse performance compared to \HopeFull and \HopeOnline.

\subsection{Max-Min Allocation}
\label{app:max-min-algo}

Here we provide a brief explanation for the heuristic allocation algorithm for a single-resource and filling ratio utilities from \cite{lien2014sequential}.  This algorithm was set-up to approximate the optimal solution to the max-min objective:
$$\max_{X \in \mathbb{R}^{n}} \Exp{\min_{i \in [n]} u(X_i, \theta_i)}$$
which aims to uniformly (across all agents) maximize the filling ratio.
Now the measure of performance for these algorithms is defined as $\Delta_{MM} = \min_{i} u(X_i^{alg}, \theta_i)$.  The \MaxMin heuristic (called Two Node Decomposition Heuristic Algorithm in \cite{lien2014sequential}) arises from a specific form of the dynamic programming solution to the two-agent problem.  In particular, the heuristic algorithm proceeds in three phases.

\begin{itemize}
    \item \textit{Decomposition}: The $n$ agent resource allocation problem is decomposed into a sequence of two agent allocation problems.
    \item \textit{Supply Allotment}:  For each two agent problem, the available budget is divided so that each portion is utilized to solve the different two-agent problems, with the rest saved for the agents who are yet to be visited.  The allotment $\hat{B^i}$ for a two agent problem for agents $i$ and $i+1$ is calculated by $$\hat{B^i} = B^i \frac{\mu_i + \mu_{i+1}}{\sum_{j=i}^n \mu_j}$$
    where $B^i$ is the current budget remaining and $\mu_i = \Exp{\theta_i}$.
    \item \textit{Resource Allocation}: For each two node problem consisting of agent $i$ and agent $i+1$ the allocation being made is a threshold policy as follows:
    \begin{align*}
        w_f^i & = \hat{B^i} \frac{\theta_i}{\theta_i + \tilde{m}_{i+1} + \delta_{i+1}\sqrt{\sigma_{i+1}}} \\
        X_i & = \min(w_f^i, \beta_{min}^{i-1}\theta_i)
    \end{align*}
    where $\tilde{m}_i = \text{Median}(\theta_i)$, $\sigma_i = \text{Variance}(\theta_i)$ and $\delta_i = \frac{\tilde{m}_i - \tilde{m}_{i+1}}{(\tilde{m}_i + \tilde{m}_{i+1})/2}$ and $\beta_{min}^i = \min_{j < i} u(X_j^{alg}, \theta_j)$ is the minimum fill rate thus far.
\end{itemize}

We use this algorithm as a benchmark in \cref{sec:experiments} since it considers a similar problem set-up and utility function but with a different objective.  We compare our algorithms to \MaxMin, including the metric $\Exp{\Delta_{MM}}$ to which this algorithm was constructed for in \cref{app:experiments}.
    \section{Omitted Proofs}
\label{app:proofs}
In this section we include the omitted proofs from the main paper.  We restate each of them here for ease of presentation.

\subsection{Approximate Fairness}

\begin{lemma}[\cref{lemma:e-close}]
Suppose that an algorithm $X^{alg}$ satisfies $\Exp{\norm{X^{opt} - X^{alg}}_{max}} \leq \epsilon$.  Then we have
\begin{itemize}
\item \textit{Approximate Envy-Freeness}:
$$\Exp{\Delta_{EF}} = \Exp{\max_{i,j} \left( u(X_j^{alg}, \theta_i) - u(X_i^{alg}, \theta_i)\right)} \leq {2L \epsilon}.$$
\item \textit{Approximate Pareto-Efficiency}: $$\Exp{\Delta_{PE}} = \frac{1}{n}\Exp{\max_k(B_k - \sum_{i} S_i X_{i,k}^{alg})} \leq \frac{S}{n}\epsilon$$
\item \textit{Approximate Proportionality}: $$\Exp{\Delta_{Prop}} = \Exp{\max_i \left( u(B/S, \theta_i) - u_i(X_i^{alg}, \theta_i)\right)} \leq L \epsilon$$
\end{itemize}
\end{lemma}
\begin{proof}

Notice that for any agent $i$, by assumption on the utility functions being $L$-Lipschitz continuous we have that $$|u(X_i^{alg}, \theta_i) - u(X_i^{opt}, \theta_i)| \leq L \norm{X_i^{opt} - X_i^{alg}}_\infty \leq L \norm{X^{opt} - X^{alg}}_{max}.$$

Each of the three properties follows then by the fact that $X^{opt}$ is the optimal fair solution in hindsight.

\medskip

\noindent \textit{Approximate Pareto-Efficiency}: Notice that $X^{opt}$ is Pareto-efficient, and so by~\cref{lem:pe-waste} we have that for any resource $k$, $B_k = \sum_{i=1}^n S_i X_{i,k}^{opt}$.  Thus we have that
\begin{align*}
    B_k - \sum_{i=1}^n S_i X_{i,k}^{alg} & = B_k - \sum_{i=1}^n S_i X_{i,k}^{opt} + \sum_{i=1}^n S_i (X_{i,k}^{alg} - X_{i,k}^{opt}) \\
    & = \sum_{i=1}^n S_i(X_{i,k}^{alg} - X_{i,k}^{opt}) \\
    & \leq \sum_{i=1}^n S_i \norm{X^{alg} - X^{opt}}_{max} \\
    & \leq S \norm{X^{alg} - X^{opt}}_{max}.
\end{align*}
Thus we get by taking the maximum over $k$ that $$\Exp{\Delta_{PE}} = \frac{1}{n} \Exp{\max_{k \in [K]} (B_k - \sum_{i=1}^n S_i X_{i,k}^{alg})} \leq \frac{S}{n} \epsilon.$$

\medskip
	 
\noindent \textit{Approximate Proportionality}:  By adding and subtracting $u(X_i^{opt}, \theta_i)$ and using the fact that $X^{opt}$ is proportional so $u(X_i^{opt}, \theta_i) \geq u(B/S, \theta_i)$ we get:
\begin{align*}
u(B/S, \theta_i) - u(X_i^{alg}, \theta_i) & = u(B/S, \theta_i) - u(X_i^{opt}, \theta_i) + u(X_i^{opt}, \theta_i) - u(X_i^{alg}, \theta_i) \\
& \leq u(X_i^{opt}, \theta_i) - u(X_i^{alg}, \theta_i) \\
& \leq L \norm{X^{opt} - X^{alg}}.
\end{align*}
Since this inequality is true almost surely for any $i$ we can take the maximum of the left hand side over $i$ and the expectation of both to show that $\Exp{\Delta_{Prop}} \leq L \epsilon$.
	 
	 \medskip
	 
\noindent \textit{Approximate Envy-Freeness}: Again by adding and subtracting and noting that for any $j$ as the optimal solution is envy-free $u(X_j^{opt}, \theta_i) \leq u(X_i^{opt}, \theta_i)$ for any $j$ and $i$ we have that
 \begin{align*}
  u(X_j^{alg}, \theta_i) - u(X_i^{alg}, \theta_i) & = u(X_j^{alg}, \theta_i) - u(X_j^{opt}, \theta_i) + u(X_j^{opt}, \theta_i) - u(X_i^{opt}, \theta_i) \\
  & \quad + u(X_i^{opt}, \theta_i) - u(X_i^{alg}, \theta_i) \\
  & \leq u(X_j^{alg}, \theta_i) - u(X_j^{opt}, \theta_i) + u(X_i^{opt}, \theta_i) - u(X_i^{alg}, \theta_i)\\
  & \leq 2L\norm{X^{alg} - X^{opt}}_{max}.
 \end{align*}
 As this is true for any $i$ and $j$, we take the maximum over all $i$ and $j$ on the left hand side and expectations to show that $\Exp{\Delta_{EF}} \leq 2L \epsilon$.
\end{proof}

\subsection{Distance to Pareto-Efficiency}

From a practical perspective, measuring the distance to Pareto-efficiency is not straightforward.  One important proxy in the context of food bank resource allocation is \textit{waste}, defined as $B - \sum_{i=1}^n S_i X_i$ under an allocation $X$.  This is the vector of leftover unallocated resources by an allocation.  
\begin{proposition}[\cref{lem:pe-waste}]
\label{applem:pe-waste}
If an allocation $X \in \mathbb{R}^{n \times K}$ is Pareto-efficient, then we have that $\sum_{i} S_i X_i = B$.
\end{proposition}
\begin{proof}
Suppose for sake of contradiction we have an allocation $X \in \mathbb{R}^{n\times k}$ which is Pareto-efficient, but there exists a resource $k$ such that $\sum_{i} S_i X_{i,k} < B_k$.  Consider a new allocation $Y$ such that $Y_{i,k} = X_{i,k} + B_k - \sum_{i} S_i X_{i,k}$ for any agent $i$ with strictly increasing utility for resource $k$.  Then $Y$ is still a valid allocation, and agent $i$ has strictly higher utility under this allocation $(\Rightarrow \Leftarrow)$
\end{proof}

An alternative expression would be to define the distance to Pareto-efficiency in terms of the maximum difference between the allocation considered and a Pareto-efficient one.  In particular we have:
\begin{definition}
Given a set of agent types $\{\theta_i\}_{i \in [n]}$ and utility functions $\{u(X, \theta)\}_{\theta \in \Theta}$ and budget $B$ we define the set of Pareto-efficient allocations as:
$$ PE(\{\theta_i\}_{i \in [n]}, B) = \{ X \in \mathbb{R}^{n \times K} \mid X \text{ is Pareto efficient}\}.$$
With this, the alternative distance to Pareto-efficiency for an allocation $X^{alg}$ is defined as $$\overline{\Delta_{PE}} = \min_{Y \in PE(\{\theta_i\}_{i \in [n]}, B)} \norm{X^{alg} - Y}_{max}.$$
\end{definition}
Note that under this definition, both the allocation $X^{alg}$ given by the algorithm is a random variable (depending on randomness in the algorithm), but also the set $PE(\{\theta_i\}_{i \in [n]}, B)$ (as its definition depends on the realized types for each agent).  With this, we notice that our definition of an $\epsilon$-optimal allocation algorithm also satisfies $\epsilon$ guarantees in terms of this alternative definition for distance to Pareto-efficiency.  In particular we have
\begin{lemma}
Suppose that the algorithm $X^{alg}$ satisfies $\Exp{\norm{X^{alg} - X^{opt}}_{max}} \leq \epsilon$.  Then we also have that $$\Exp{\overline{\Delta_{PE}}} \leq \epsilon.$$ 
\end{lemma}
\begin{proof}
This follows from noting that the optimal solution $X^{opt}$ is Pareto efficient based on the realized types and so belongs in $PE(\{\theta_i\}_{i \in [n]}, B)$.  Thus we have that $\overline{\Delta_{PE}} \leq \norm{X^{alg} - X^{opt}}$.  The result follows from taking the expectation of both sides.
\end{proof}

\subsection{Filling-Ratio Utilities}

In this section we specialize the previous results to the single-resource case with filling ratio utilities $u(X, \theta) = \min(\frac{X}{\theta}, 1)$.  Unfortunately, as these utility functions are not strictly increasing many of the previous proofs do not directly follow under these utility functions.  However, we show respective definitions and results which follow from efficient solutions to the Eisenberg-Gale program with homogeneous concave utility functions~\cite{roughgarden2010algorithmic}.

\begin{theorem}[\cref{thm:single-resource_app}]
\label{thm:app_offline}
An optimal solution to~\cref{eq:offline_nsw} for a fixed set of agent demands $\{\theta_i\}_{i \in [n]}$ is given by a waterfilling threshold solution, where the allocation $$X_i^{opt} = \min\{w_f, \theta_i, B\}$$ and the waterfilling threshold $w_f$ solves $$\min\left(B, \sum_{i=1}^n S_i \theta_i\right) = \sum_{i=1}^n S_i \theta_i \Ind{\theta_i \leq w_f} + S_i w_f \Ind{\theta_i \geq w_f}.$$
Moreover, this allocation is Pareto-efficient, envy-free, proportional, and hence fair.
\end{theorem}
\begin{proof}
First we show that the optimal solution to~\cref{eq:offline_nsw} is the waterfilling solution by showing it satisfies the KKT conditions, before proceeding to show that the solution is also fair.

Let $X_i^{opt}$ be the proposed optimal solution.  In taking the Lagrangian of the optimization problem we introduce dual variables $\lambda$ for the constraint that $X \geq 0$ and $\mu$ for the constraint that $\sum_{i=1}^n S_i X_i \leq B$.  This yields the following Lagrangian optimization problem:
$$\max_{X \in \mathbb{R}^{n}} \min_{\lambda \in \mathbb{R}^n, \mu \in \mathbb{R}} - \sum_{i=1}^n S_i \log\left( \min \left(\frac{X_i}{\theta_i}, 1 \right) \right) - \lambda^\top X + \mu(B - \sum_{i=1}^n S_i X_i).
$$
Taking the sub-gradient and using the complementary slackness conditions we get the following conditions:
\allowdisplaybreaks
\begin{align*}
X & \geq 0 & \lambda & \geq 0 \\
\mu & \geq 0 & \sum_{i=1}^n S_i X_i & \leq B \\
\lambda_i X_i & = 0 & \mu(B - \sum_{i=1}^n S_i X_i) & = 0 \\
0 & \in - S_i \partial \log\left( \min \left\{\frac{x_i}{D_i}, 1\right\}\right) - \lambda_i + S_i \mu.
\end{align*}
Since $X_i^{opt} \neq 0$ we can safely set the dual variables $\lambda_i = 0$.  Moreover, we also have that the subgradient of the log utility term is:
$$\partial \log\left( \min \left(\frac{X_i}{\theta_i}, 1\right)\right) = \begin{cases}
0 \text{ if } X_i > \theta_i\\
\left[0, \frac{1}{\theta_i} \right] \text{ if } X_i = \theta_i\\
\frac{1}{X_i} \text{ if } X_i < \theta_i
\end{cases}$$
To summarize we have the following conditions after eliminating the dual variables $\lambda$:
\begin{align*}
X & \geq 0 & \mu & \geq 0\\
\sum_{i=1}^n S_i X_i & \leq B & \mu(B - \sum_{i=1}^n S_i X_i) & = 0 \\
\mu & \in \partial \log\left( \min \left(\frac{X_i}{\theta_i}, 1\right)\right)
\end{align*}

First consider the case when the dual variable $\mu = 0$.  Then by the subgradient condition we must have that $X_i \geq \theta_i$ for every $i$, and by feasability that $\sum_{i=1}^n X_i \leq B$.  The waterfilling solution will then have a waterfilling level $w_f = \max_{i = 1, \ldots, n} \theta_i$ and the allocation will be $X_i^{opt} = \theta_i$.  This solution satisfies all of the conditions listed, and so will be optimal.

For the case when $\mu \neq 0$ then we must have that $\sum_{i=1}^n S_i X_i = B$ by complementary slackness.  Moreover, the gradient condition asserts that $X_i \leq \theta_i$ for every $i$.  Without loss of generality we will assume that $\theta_1 \leq \theta_2 \leq \ldots \leq \theta_n$ and break into cases:

\medskip \noindent \textbf{Case I}: $\theta_1 \geq \frac{B}{S}.$

\smallskip

In this case the optimal solution $X_i^{opt} = \frac{B}{S} = w_f$ is the waterfilling solution.  This is as we can set $\mu = \frac{S}{B}$ and check that all of the KKT conditions hold.

\medskip \noindent \textbf{Case II}: $\theta_n < \frac{B}{S}$.

This is impossible as we must have that $\sum_{i=1}^n S_i X_i = B \leq \sum_{i=1}^n S_i \theta_i < B$.

\medskip \noindent \textbf{Case III}: Otherwise let $\mu = \frac{1}{w_f}$ where $w_f$ is the waterfilling level.  Then for any agent $i$ such that $X_i^{opt} = \theta_i$ then $\mu = \frac{1}{w_f} \leq \frac{1}{\theta_i}$.  Moreover, for any agent such that $X_i^{opt} < \theta_i$ then the allocation is $X_i = w_f$ and so $\mu = \frac{1}{X_i} = \frac{1}{w_f}$.  Noting that this mimics the subgradient condition we see that the waterfilling solution satisfies the complementary slackness conditions and so is optimal.

\medskip

Next we show that the waterfilling solution is envy-free, proportional, and Pareto-efficient.  We start by showing that $w_f \geq \frac{B}{S}$.  Indeed,
\begin{align*}
	w_f & = w_f \Ind{\theta_i \leq w_f} + w_f \Ind{\theta_i \geq w_f} \\
	& \geq \theta_i \Ind{\theta_i \leq w_f} + w_f \Ind{\theta_i \geq w_f}.
\end{align*}
Summing up from $i = 1, \ldots, n$ and multiplying the terms by $S_i$ we find that
\begin{align*}
	S w_f & \geq \sum_{i=1}^n S_i \theta_i \Ind{\theta_i \leq w_f} + S_i w_f \Ind{\theta_i \geq w_f} \\
		& = B.
\end{align*}

\medskip \noindent \textit{Pareto-Efficient}: We first show that $\sum_{i=1}^n S_i X_i^{opt} = \min(B, \sum_{i=1}^n S_i \theta_i)$.  Indeed,
\begin{align*}
    \sum_{i=1}^n S_i X_i^{opt} & = \sum_{i=1}^n S_i \min(w_f, \theta_i) \\
    & = \sum_{i=1}^n S_i w_f \Ind{w_f \leq \theta_i} + S_i \theta_i \Ind{\theta_i \leq w_F} \\
    & = \min(B, \sum_{i=1}^n S_i \theta_i)
\end{align*}
Now suppose for the sake of contradiction that there exists some other allocation $Y \in \mathbb{R}^{n}$ such that $u(X_i^{opt}, \theta_i) < u(Y_i, \theta_i)$ and $u(X_j^{opt}, \theta_j) \leq u(Y_j, \theta_j)$ for every $j \neq i$.  Then we must have that $X_i^{opt} < \theta_i$ as otherwise the utilities would both be one, and so $X_i^{opt} = w_f$.  Moreover, by definition of the waterfilling level this implies that $\sum_{i=1}^n S_i \theta_i > B$ and so $\sum_{i=1}^n S_i X_i^{opt} = B$.

Consider any agent $j \neq i$, as $u(Y_j, \theta_j) \geq u(X_j^{opt}, \theta_j)$ we must have that $Y_j \geq X_j^{opt}$ as the utilities are increasing.  Hence we see that
\begin{align*}
    B & = \sum_{j=1}^n S_j X_j^{opt} \\
    & < \sum_{j=1}^n S_j Y_j
\end{align*}
which contradicts $Y$ being a valid allocation $(\Rightarrow \Leftarrow)$.

\medskip

\noindent \textit{Proportional}:  Note that for any group $i$ such that $X_i^{opt} = \theta_i$ then
\begin{align*}
	u(X_i^{opt}, \theta_i) & = \min \left( \frac{\theta_i}{\theta_i}, 1 \right) = 1 \\
	& \geq \min \left(\frac{B}{S \theta_i}, 1 \right) = u(B/S, \theta_i).
\end{align*}
Similarly if $X_i^{opt} = w_f$ then $w_f \leq \theta_i$ so
\begin{align*}
	u(X_i^{opt}, \theta_i) & = \min \left( \frac{w_f}{\theta_i}, 1\right)  = \frac{w_f}{\theta_i} \geq \frac{B}{S \theta_i} \\
	& = u(B/S, \theta_i).
\end{align*}

\medskip

\noindent \textit{Envy-Free}:  Consider a group $i$ and let $j$ be any other group.  If $X_i^{opt} = \theta_i$ then
\[u(X_i^{opt}, \theta_i) = 1 \geq u(X_j^{opt}, \theta_i)\]
so the agent is trivially envy free.  Otherwise, if $X_i^{opt} = w_f$ then for any group $j$ either $X_j^{opt} = \theta_j$ or $X_j^{opt} = w_f$.  If $X_j^{opt} = w_f$ then clearly $u(X_i^{opt}, \theta_i) = u(X_j^{opt}, \theta_i)$.  However, if $X_j^{opt} = \theta_j$ then it must be true that $\theta_j \leq w_f$ and so
\begin{align*}
u(X_i^{opt}, \theta_i) & = \frac{w_f}{\theta_i} \geq \frac{X_j^{opt}}{\theta_i} \\
& = u(X_j^{opt}, \theta_i).
\end{align*}
\end{proof}

    \begin{table}[!t]
\caption{Normalized sizes / mean demands for the different counties in the Southern Tier of New York and their 2019 population.}
\label{tab:pop_sizes}
\setlength\tabcolsep{0pt} 
\centering
\begin{tabular*}{\columnwidth }{@{\extracolsep{\fill}}rcccccc}
\toprule
County & Broome & Steuben & Chemung & Tioga & Schuyler & Tompkins \\
\midrule
Size $S_i$ & 26.72 & 34.55 & 12.09 & 12.35 & 2.96 & 11.31 \\
Population & 190,488 & 95,379 & 83,456 & 48,203 & 17,912 & 102,180\\
\bottomrule
\end{tabular*}
\end{table}

\begin{table}[!t ]
\caption{Weights $w_k$ for the different products considered.}
\label{tab:weights}
\setlength\tabcolsep{0pt} 
\centering
\begin{tabular*}{\columnwidth }{@{\extracolsep{\fill}}rccccccccc}
\toprule
Resource & Cereal & Diapers & Pasta & Paper & Prepared Meals & Rice & Meat & Fruit & Produce \\
\midrule
Weight $w_k$ & 3.9 & 3.5 & 3.2 & 3 & 2.8 & 2.7 & 1.9 & 1.2 & 0.2 \\
\bottomrule
\end{tabular*}
\end{table}

\section{Full Experimental Results}
\label{app:experiments}

Here we provide a discussion on all of the simulations conducted.  Some of this will be a repeat of \cref{sec:experiments} while adding measures of variance of the results.  For ease of presentation we include the same benchmark algorithms discussed previously and include further discussion and other heuristic algorithms in the attached code base.  Moreover, all figures are deferred until after the discussion to save on space.  Each simulation uses data arising from our motivation in resource allocation for Food Banks, and we compare the algorithms on the approximate fairness definitions from \cref{def:distance} and \cref{def:online-fairness}.

\subsection{Single Resource}
\label{app:si ngle_resource}

We start by discussing the single-resource variant of the food bank allocation problem motivated in~\cref{section: examples}.  In this setting, a mobile food pantry loads up the truck at the start of the day with a fixed number of `meals' $B$, and travels sequentially from one drop-off location to the next.  At each location $i$, they observe a demand $\theta_i \in \mathbb{R}_+$ drawn from a known distribution $\F_i$, make an allocation $X_i^{alg}$, before proceeding to the next drop off point.
We consider the filling ratio utilities $u(X, \theta) = \min(\frac{X}{\theta}, 1)$.  These utilities are of particular importance to food-banks for several reasons:
\begin{itemize}
    \item Feeding America collects millions of pounds of food donations across the United States and uses a centralized allocation mechanism to redistribute these resources to food banks across the country.  As a proxy for money, each food bank is given a daily budget to use in the auction based on their `Goal Score'.  Major components of this score include population, local supply, but also efficiency, i.e. the relative demand the food bank is able to satisfy for its distribution sites~\cite{prendergast2017food}.  Filling ratio utilities serve as a proxy for ensuring efficiency across different drop-off locations and distribution sites.
    \item The utility functions are normalized by the relative demand of the different locations, placing distribution sites of varying sizes at equal levels.
\end{itemize}

We start by discussing the structural result.  While the filling ratio utilities chosen are not monotonically increasing (due to the minimum), the Eisenberg-Gale program in \cref{eq:offline_nsw} still guarantees a fair allocation.  Moreover, the optimal solution can be characterized via a Waterfilling solution (seen in \cref{fig:allocation}).

\begin{theorem}
\label{thm:single-resource_app}
An optimal solution to~\cref{eq:offline_nsw} for a fixed set of agent demands $\{\theta_i\}_{i \in [n]}$ is given by a waterfilling policy
$$X_i^{opt} = \min\{w_f, \theta_i, B\},$$ where the waterfilling threshold $w_f$ solves $$\min\left(B, \sum_{i=1}^n S_i \theta_i\right) = \sum_{i=1}^n S_i \theta_i \Ind{\theta_i \leq w_f} + S_i w_f \Ind{\theta_i \geq w_f}.$$
Moreover, this allocation is Pareto-efficient, envy-free, proportional, and hence fair.
\end{theorem}

Important to note, is that the allocation algorithms discussed, (\HopeOnline, \HopeFull, \EtOnline, \EtFull) all make allocation decisions based on a formulation of the Eisenberg-Gale where unknown quantities are replaced with their histogram or expectation.  A slight caveat is that \HopeOnline and \HopeFull solve optimization problems where the summation is over agent types instead of over agents.  For each of the simulation results, we include an additional plot of the estimated waterfilling level for a specific agent, i.e. $w_f^i$, which is the waterfilling threshold the optimization problem for the algorithm returns when visiting agent $i$.  As an example, the \HopeOnline algorithm allocates according to the solution to:
\begin{align*}
\max_{X \in \mathbb{R}_+^{|\Theta| \times K}} & \frac{1}{S} \sum_{\theta \in \Theta} N_\theta \log(u(X_\theta, \theta)) \\
\text{s.t. } & \sum_{\theta \in \Theta} N_\theta X_\theta \leq B^i.
\end{align*}
where $B^i = B^{i-1} - X_{i-1}^{alg}$ is the current available resources, and the expected histogram over types is defined as 
\begin{align*}
N_\theta = S_i \Ind{\theta_i = \theta} + \sum_{j=i+1}^n S_j \Pr(\theta_j = \theta).
\end{align*}

Thus, the waterfilling threshold at agent $i$, $w_f^i$ will be the solution to:
$$\min \left( B^i, \sum_{\theta \in \Theta} N_\theta \theta \right) = \sum_{\theta \in \Theta} N_\theta \theta \Ind{\theta \leq w_f^i} + N_\theta w_f \Ind{\theta \geq w_f^i}.$$

\bigskip \noindent \textbf{Heuristics}: We compare the following algorithms:
\begin{itemize}
    \item \HopeOnline (see \cref{sec:approximation_algorithms})
    \item \HopeFull (see \cref{sec:approximation_algorithms})
    \item \EtOnline (see \cref{app:et-algo})
    \item \EtFull (see \cref{app:et-algo})
    \item \MaxMin (see \cref{app:max-min-algo})
    \item \textsc{Greedy}: where every agent is given its demand $\theta_i$ until the budget is depleted
    \item \textsc{Adaptive-Threshold}: where the allocation $X_i^{alg} = \min(B^i / (n-i), \theta_i)$ provides either an agents demand or an equal share of the remaining budget.
\end{itemize}

\medskip \noindent \textbf{Choice of Demand Distributions}: We perform several synthetic experiments, with demand parameters based on food bank demand data.  In each simulation we set the size $S_i = 1$ for each agent, and sample the demands $\theta_i$ as follows:
\begin{itemize}
    \item \emph{FBST Dataset}: Here we consider the setting with $n = 6$ agents corresponding to the six counties serviced by the Food Bank of the Southern Tier of New York \cite{fbst}.  For the type distributions, we set $\theta_i \sim \F_i$ where $\F_i$ is a discretized Gaussian with mean and variance based on historical food demands.  We normalized the means of the distributions for them to total to one hundred (shown in \cref{tab:pop_sizes}).
    \item \emph{Gaussian Demands}: Here we set the type distribution $\theta_i \sim \F_i$ as an i.i.d. discretized Gaussian distribution with mean fifteen and variance three, where we discretized the distribution into twenty buckets.
    \item \emph{Poisson Demands}: Here we set the type distribution $\theta_i \sim \F_i$ as an i.i.d. discretized Poisson distribution with $\lambda = 10$.  To discretize the distribution to have finite support we set the mass of all points larger than twenty to be zero, and redistributed that mass to $\Pr(\theta_i = 1)$.
    \item \emph{Simple Distribution}: Here we set the type distribution $\theta_i \sim \F_i$ where $\F_i = \text{Uniform}\{1,2\}$.
\end{itemize}

\medskip \noindent \textbf{Metrics Included}:  We include plots the simulation results in Figures~\ref{fig:gaussian}, \ref{fig:poisson}, and \ref{fig:simple}, and table of the fairness metrics in Tables~\ref{tab:gaussian}, \ref{tab:poisson}, \ref{tab:simple}, and \ref{tab:fbst}.  In the figures we include four plots of the following:
\begin{itemize}
    \item $\Exp{\norm{X^{opt} - X^{alg}}_{max}}$, the expected maximum difference between allocations as we scale the number of agents $n$ from $1$ to $100$
    \item $\Exp{\norm{X^{opt} - X^{alg}}_{max}}$, the expected additive difference between allocations as we scale the number of agents $n$ from $1$ to $100$
    \item $\Exp{w_f^i}$, the expected waterfilling threshold of the algorithms on agent $i$ with $n = 100$ agents total
    \item $\Exp{|X_i^{alg} - X^{opt}_i|}$, the expected difference in allocations for agent $i$ with $n = 100$ agents total
\end{itemize}

In the tables we include:
\begin{itemize}
    \item $\Exp{\norm{X^{opt} - X^{alg}}_{max}}$, the expected maximum difference between the allocation generated by the algorithm and the fair one in hindsight (\cref{def:online-fairness})
    \item $\Exp{\norm{X^{opt} - X^{alg}}_{1}}$, the expected additive difference between the allocation generated by the algorithm and the fair one in hindsight
    \item $\Exp{\Delta_{EF}}$, the expected maximum envy between any two agents (\cref{def:distance})
    \item $\Exp{\Delta_{PE}}$, the expected waste (\cref{def:distance})
    \item $\Exp{\Delta_{Prop}}$, the expected maximum envy between an agent and equal allocation (\cref{def:distance})
    \item $\Exp{\Delta_{MM}} = \Exp{\min_i u(X_i^{alg}, \theta_i)}$, the expected minimum fill rate ~\cite{lien2014sequential}
\end{itemize}

\medskip \noindent \textbf{Summary of Results}: Here we provide a brief discussion on the major metrics and plots.

\smallskip \noindent \textit{Scaling with $n$}: From the plots (figures~\ref{fig:gaussian}, \ref{fig:poisson}, and \ref{fig:simple}) we see that \HopeOnline performs competitively across all metrics.  In particular, from the plots in the top-left showing $\Exp{\norm{X^{alg} - X^{opt}}_{max}}$ we see that \HopeOnline has constant scaling with respect to the number of agents $n$ in comparison to the other heuristic algorithms.  Similar results are shown in the plots in the top right, where we see $\Exp{\norm{X^{alg} - X^{opt}}_1}$.  In these plots, \HopeOnline has sublinear dependence with respect to the number of agents $n$.  This supports \HopeOnline as being a promising candidate for achieving approximate fairness (\cref{def:online-fairness}).

Moreover, the \MaxMin algorithm often has linear dependence for the $\ell_1$ norm.  This arises as the algorithm under allocates at every time period (shown in the plots in the bottom right).  This is due to the formulation of the algorithm to provide guarantees in terms of the max-min objective, where once a mistake has been made the algorithm has no incentive to correct it.  This is typical for allocation algorithms formulated under the max-min objective, and motivates other fairness objectives (such as \cref{def:distance}).

Lastly, in the plot on the bottom left we see that \HopeOnline uses a waterfilling threshold that closely matches the true one (shown in black in the figures).  The figures on the bottom right shows that \HopeOnline provides a uniform approximation to the optimal fair allocation in hindsight $X^{opt}$, with a slight deviation for later arriving agents.  This is in contrast to other heuristics, which often have extreme suboptimal performance for later agents (as in the case of \textsc{Greedy}), or uniformly underallocates (as in the case of \MaxMin).
Lastly we see that \HopeOnline has a waterfilling level that closely matches the true one, and that the algorithm is uniformly close to $X^{opt}$ across each agent.   

\smallskip \noindent \textit{Performance on fairness distance} (\cref{def:distance} and \cref{def:online-fairness}): In the tables we see that \HopeOnline either has the best, or second best performance across all of the fairness metrics considered.  Important to note, is that \HopeOnline is also competitive in terms of $\Delta_{MM}$, an objective that \MaxMin was formulated to perform with respect to.

\subsection{Multiple Resource}
\label{app:multi_resource}

Here we consider the multiple resource allocation problem with linear utility functions $u(x, \theta) = \langle x, \theta \rangle$.  Now, the agent type $\theta \in \mathbb{R}^{k}$ represents a vector of preferences, where $\theta_k$ is the agent's preference for resource $k$.  We consider the resource allocation problems with $n = 6$ agents, corresponding to the six counties serviced by the Food bank of the Southern Tier of New York.  The sizes $S_i$ are taken as representative of their total 2018 food demand normalized to sum up to one hundred, displayed in \cref{tab:pop_sizes}~\cite{fbst}.  We also include their respective populations, highlighting the choice of using their normalized food demands as a representative of their size instead of their population, as different counties have different food assistance needs irrespective of population.

\medskip \noindent \textbf{Heuristics}: We compare the following algorithms
\begin{itemize}
    \item \HopeOnline (see \cref{sec:approximation_algorithms})
    \item \HopeFull (see \cref{sec:approximation_algorithms})
    \item \EtOnline (see \cref{app:et-algo})
    \item \EtFull (see \cref{app:et-algo})
\end{itemize}

\medskip \noindent \textbf{Choice of Demand Distributions}: We consider $k = 9$ resources.  To generate the preference distributions, we created eight different preference profiles $\Theta = \{\theta^1, \ldots, \theta^8\}$, where each component $\theta_k^i = \text{Bernoulli}(1/2) w_k$ where $w_k$ is the price of the product used in the non-monetary auctions Feeding America uses to distribute resources across the United States displayed in \cref{tab:weights}~\cite{prendergast2017food}.  We sampled eight such type vectors, and then considered the uniform distribution over those eight types for each distribution $\F_i$, i.e. $\F_i = \text{Uniform}(\Theta)$.

\medskip \noindent \textbf{Metrics Included}:  As the number of agents $n = 6$ is fixed we only include a table summarizing the metrics.  In the table we include:
\begin{itemize}
    \item $\Exp{\norm{X^{opt} - X^{alg}}_{max}}$, the expected maximum difference between the allocation generated by the algorithm and the fair one in hindsight (\cref{def:online-fairness})
    \item $\Exp{\norm{X^{opt} - X^{alg}}_{1}}$, the expected additive difference between the allocation generated by the algorithm and the fair one in hindsight
    \item $\Exp{\Delta_{Util}} = \Exp{\max_i |u_i(X_i^{opt}, \theta_i) - u_i(X_i^{alg}, \theta_i)|}$, the expected maximum difference in utility an agent receives between the allocation generated by the algorithm and the fair one in hindsight
    \item $\Exp{\Delta_{EF}}$, the expected maximum envy between any two agents (\cref{def:distance})
    \item $\Exp{\Delta_{PE}}$, the expected waste (\cref{def:distance})
    \item $\Exp{\Delta_{Prop}}$, the expected maximum envy between an agent and equal allocation (\cref{def:distance})
\end{itemize}

\medskip \noindent \textbf{Summary of Results}: In \cref{tab:mult-fairness_app} we compare the fairness metrics averaged over one thousand simulations with the addition of a standard normal confidence interval.  Here we see that \HopeOnline is competitive with respect to all metrics in comparison to the other heuristic algorithms.  In all metrics, \HopeOnline either performs the best or second-best.  More important, is that \HopeOnline performs the best in terms of minimizing the distance to the optimal allocation in hindsight, $\Exp{\norm{X^{opt} - X^{alg}}_{max}}$.

\begin{table*}[!ht]
\caption{ 
Comparison of fairness metrics averaged over 1000 simulations on the multiple-resource online allocation problem with linear utilities.  We plot the four metrics from \cref{def:online-fairness} and \cref{def:distance}.  Larger values corresponds to a lower score on that metric.  We include both the mean and a standard normal approximation confidence interval for the metrics.  Due to space constraints in the table, we include a separate row with the order of magnitude of the confidence intervals for these results.} \label{tab:mult-fairness_app}
\setlength\tabcolsep{0pt} 
\footnotesize\centering
\begin{tabular*}{\columnwidth }{@{\extracolsep{\fill}}r*6c}
\toprule
Algorithm &  {$\Exp{\norm{X^{opt} - X^{alg}}_{max}}$} & $\Exp{\norm{X^{opt} - X^{alg}}_1}$ & $\Exp{\Delta_{Util}}$ & {$\Exp{\Delta_{EF}}$} & {$\Exp{\Delta_{PE}}$} & {$\Exp{\Delta_{Prop}}$}\\
\midrule  
Level of confidence & $10^{-5}$ & $10^{-5}$ & $10^{-4}$ & $10^{-4}$ & $10^{-5}$ & $10^{-4}$ \\
\midrule
\HopeOnline   &  0.0006 & \textbf{0.0054} & \textbf{0.0015} & \textbf{0.0015} & \textbf{0.00026} & \textbf{0.0010} \\
\HopeFull   &  \textbf{0.00042} & 0.0070 & 0.0018 & 0.016 & 0.0018 & 0.0011 \\
\EtOnline   &  0.0013 & 0.0080 & 0.0016 & 0.016 & 0.0071 & 0.0011 \\
\EtFull   &  0.00070 & 0.0067 & 0.0029 & 0.024 & 0.0018 & 0.0023 \\
\bottomrule
\end{tabular*}
\end{table*}

\subsection{Experiment Setup and Computing Infrastructure}

\medskip \noindent \textbf{Experiment Setup}: Each experiment was run with $1000$ iterations where the relevant plots are taking the mean and a standard-normal confidence interval of the related quantities.  In the case of a single resource the budget $B$ is set to be the total expected demand.  For the experiments with multiple resources we use a budget of the total preferences for each product as the allocations are the same up to scaling.  All randomness is dictated by a seed set at the start of each simulation for verifying results.

\medskip \noindent \textbf{Computing Infrastructure}: The experiments were conducted on a personal computer with an AMD Ryzen 5 3600 6-Core 3.60 GHz processor and 16.0GB of RAM. No GPUs were harmed in these experiments. 

\medskip \noindent \textbf{Run-time Analysis}: The average computation time for running a single iteration of the heuristic algorithms in comparison to the offline solution is listed in~\cref{tab:computation}.  These statistics were computed by averaging over $1000$ simulations.  The case with $n = 6$ is on the multiple resource allocation dataset described in \cref{app:multi_resource}.  The case with $n = 100$ is on the single resource allocation dataset with an i.i.d. discretized Gaussian from \cref{app:si ngle_resource}.
These results mostly serve to indicate how the algorithms scale well and are easily implementable.

\begin{table}[!h]
\caption{Comparison of the running time (in seconds) for calculating the allocations used in the five main heuristics used and the offline optimal solution.}
\label{tab:computation}
\setlength\tabcolsep{0pt} 
\centering
\begin{tabular*}{\columnwidth }{@{\extracolsep{\fill}}rcccccc}
\toprule
Algorithm &  \textsc{Offline} & \HopeOnline & \HopeFull & \EtOnline & \EtFull & \MaxMin \\
\midrule
$n = 6$ & \textbf{0.00031} & 0.0035 & 0.0025 & 0.0015 & 0.00080 & N/A \\
$n = 100$ & \textbf{0.00015} & 0.0031 & 0.0029 & 0.0062 & 0.011 & 0.0010 \\
\bottomrule
\end{tabular*}
\end{table}

\begin{table*}[!ht]
\caption{ 
Comparison of fairness metrics (averaged over 1000 replications) on the single-resource online allocation problem with filling-ratio utilities on the \emph{Gaussian} problem. We compare the four unfairness metrics from~\cref{def:distance,def:online-fairness} (larger values correspond to lower scores; best value highlighted) with the addition of $\Exp{\Delta_{MM}} = \Exp{\min_i u(X_i^{alg}, \theta_i)}$, the minimum fill rate, and $\Exp{\norm{X^{alg} - X^{opt}}_1}$, the $\ell_1$ difference in allocations.  Due to space constraints in the table, we include a separate row with the order of magnitude for the confidence intervals for these results.
} \label{tab:gaussian}
\setlength\tabcolsep{0pt} 
\footnotesize\centering
\begin{tabular*}{\columnwidth }{@{\extracolsep{\fill}}r*6c}
\toprule
Algorithm & $\Exp{\Delta_{EF}}$ & $\Exp{\Delta_{PE}}$ & $\Exp{\Delta_{Prop}}$ & $\Exp{\Delta_{MM}}$ & $\Exp{\norm{X^{opt} - X^{alg}}_{max}}$ & $\Exp{\norm{X^{opt} - X^{alg}}_1}$ \\
\midrule
Level of confidence & $10^{-2}$ & $10^{-3}$ & $10^{-4}$ & $10^{-2}$ & $10^{-1}$ & $10^{-1}$ \\
\midrule
\HopeOnline  &  \textbf{0.11} & 0.14 & 0.0010 & \textbf{0.86} & \textbf{2.22} & \textbf{12.14} \\
\HopeFull & 0.25 & 0.15 & 0.020 & 0.72 & 4.51 & 12.26 \\
\EtOnline & 0.20 & 0.22 & \textbf{0.0001} & 0.78 & 3.96 & 24.36 \\
\EtFull & 0.18 & 0.24 & 0.00049 & 0.78 & 3.69 & 16.99 \\
\MaxMin & 0.13 & 2.01 & 0.0013 & 0.85 & 2.94 & 193.21 \\
\textsc{Greedy} & 0.38 & \textbf{0.11} & 0.035 & 0.61 & 6.41 & 23.91 \\
\textsc{Adaptive-Threshold} & 0.21 & 0.71 & 0.0081 & 0.71 & 5.20 & 62.24 \\
\bottomrule
\end{tabular*}
\end{table*}

\begin{table*}[!ht]
\caption{ \footnotesize 
Comparison of fairness metrics (averaged over 1000 replications) on the single-resource online allocation problem with filling-ratio utilities on the \emph{Poisson} problem. We compare the four unfairness metrics from~\cref{def:distance,def:online-fairness} (larger values correspond to lower scores; best value highlighted) with the addition of $\Exp{\Delta_{MM}} =  \Exp{\min_i u(X_i^{alg}, \theta_i)}$, the minimum fill rate, and $\Exp{\norm{X^{alg} - X^{opt}}_1}$, the $\ell_1$ difference in allocations.  Due to space constraints in the table, we include a separate row with the order of magnitude for the confidence intervals for these results.
} \label{tab:poisson}
\setlength\tabcolsep{0pt} 
\footnotesize\centering
\begin{tabular*}{\columnwidth }{@{\extracolsep{\fill}}r*6c}
\toprule
Algorithm & $\Exp{\Delta_{EF}}$ & $\Exp{\Delta_{PE}}$ & $\Exp{\Delta_{Prop}}$ & $\Exp{\Delta_{MM}}$ & $\Exp{\norm{X^{opt} - X^{alg}}_{max}}$ & $\Exp{\norm{X^{opt} - X^{alg}}_1}$ \\
\midrule
Level of confidence & $10^{-3}$ & $10^{-2}$ & $10^{-3}$ & $10^{-3}$ & $10^{-2}$ & $10^{-1}$\\
\midrule
\HopeOnline   & \textbf{0.11} & 0.14 & 0.011 & \textbf{0.86} & \textbf{2.23} & \textbf{12.14}  \\
\HopeFull & 0.27 & 0.13 & 0.021 & 0.71 & 4.76 & 12.23 \\
\EtOnline & 0.20 & 0.22 & 0.071 & 0.79 & 3.87 & 23.98  \\
\EtFull & 0.20 & 0.22 & 0.070 & 0.76 & 3.93 & 16.48  \\
\MaxMin & 0.22 & 2.32 & 0.22 & 0.77 & 3.90 & 221.73  \\
\textsc{Greedy} & 0.40 & \textbf{0.11} & 0.37 & 0.59 & 6.74 & 23.40 \\
\textsc{Adaptive-Threshold} & 0.21 & 0.70 & \textbf{0} & 0.71 & 5.17 & 61.67  \\
\bottomrule
\end{tabular*}
\end{table*}

\begin{table*}[!ht]
\caption{ \footnotesize 
Comparison of fairness metrics (averaged over 1000 replications) on the single-resource online allocation problem with filling-ratio utilities on the \emph{Simple Distribution} problem. We compare the four unfairness metrics from~\cref{def:distance,def:online-fairness} (larger values correspond to lower scores; best value highlighted) with the addition of $\Exp{\Delta_{MM}} =  \Exp{\min_i u(X_i^{alg}, \theta_i)}$, the minimum fill rate, and $\Exp{\norm{X^{alg} - X^{opt}}_1}$, the $\ell_1$ difference in allocations.  Due to space constraints in the table, we include a separate row with the order of magnitude for the confidence intervals for these results.
} \label{tab:simple}
\setlength\tabcolsep{0pt} 
\footnotesize\centering
\begin{tabular*}{\columnwidth }{@{\extracolsep{\fill}}r*6c}
\toprule
Algorithm & $\Exp{\Delta_{EF}}$ & $\Exp{\Delta_{PE}}$ & $\Exp{\Delta_{Prop}}$ & $\Exp{\Delta_{MM}}$ & $\Exp{\norm{X^{opt} - X^{alg}}_{max}}$ & $\Exp{\norm{X^{opt} - X^{alg}}_1}$ \\
\midrule
Level of confidence & $10^{-3}$ & $10^{-3}$ & $10^{-3}$ & $10^{-3}$ & $10^{-2}$ & $10^{-1}$ \\
\midrule
\HopeOnline   &  0.11 & 0.022 & 0.013 & 0.88 & 0.20 & \textbf{2.00} \\
\HopeFull & 0.32 & 0.024 & 0.26 & 0.67 & 0.50 & 2.00 \\
\EtOnline & 0.23 & 0.037 & 0.080 & 0.76 & 0.43 & 4.67 \\
\EtFull & 0.24 & 0.40 & 0.080 & 0.75 & 0.42 & 2.93 \\
\MaxMin & \textbf{0.081} & 0.031 & 0.064 & \textbf{0.91} & \textbf{0.13} & 2.84\\
\textsc{Greedy} & 0.44 & \textbf{0.019} & 0.41 & 0.55 & 0.75 & 3.71 \\
\textsc{Adaptive-Threshold} & 0.24 & 0.12 & \textbf{0} & 0.75 & 0.45 & 11.64 \\
\bottomrule
\end{tabular*}
\end{table*}

\begin{table*}[!ht]
\caption{ \footnotesize 
Comparison of fairness metrics (averaged over 1000 replications) on the single-resource online allocation problem with filling-ratio utilities on the \emph{FBST Dataset} problem. We compare the four unfairness metrics from~\cref{def:distance,def:online-fairness} (larger values correspond to lower scores; best value highlighted) with the addition of $\Exp{\Delta_{MM}} = \Exp{\min_i u(X_i^{alg}, \theta_i)}$, the minimum fill rate, and $\Exp{\norm{X^{alg} - X^{opt}}_1}$, the $\ell_1$ difference in allocations.  Due to space constraints in the table, we include a separate row with the order of magnitude for the confidence intervals for these results.
} \label{tab:fbst}
\setlength\tabcolsep{0pt} 
\footnotesize\centering
\begin{tabular*}{\columnwidth }{@{\extracolsep{\fill}}r*6c}
\toprule
Algorithm & $\Exp{\Delta_{EF}}$ & $\Exp{\Delta_{PE}}$ & $\Exp{\Delta_{Prop}}$ & $\Exp{\Delta_{MM}}$ & $\Exp{\norm{X^{opt} - X^{alg}}_{max}}$ & $\Exp{\norm{X^{opt} - X^{alg}}_1}$ \\
\midrule
Level of confidence & $10^{-3}$ & $10^{-2}$ & $10^{-3}$ & $10^{-3}$ & $10^{-1}$ & $10^{-1}$ \\
\midrule
\HopeOnline   &  0.058 & 0.36 & \textbf{0.057} & 0.92 & \textbf{1.37} & \textbf{2.24} \\
\HopeFull & 0.079 & 0.37 & 0.079 & 0.90 & 1.42 & 2.61 \\
\EtOnline & 0.091 & 0.49 & 0.89 & 0.90 & 1.68 & 3.12 \\
\EtFull & 0.079 & 0.37 & 0.079 & 0.90 & 1.38 & 3.12 \\
\MaxMin & 0.065 & 0.64 & 0.34 & \textbf{0.93} & 1.62 & 3.48 \\
\textsc{Greedy} & 0.13 & \textbf{0.32} & 0.13 & 0.86 & 1.77 & 3.54 \\
\textsc{Adaptive-Threshold} & \textbf{0.00058} & 4.68 & 0.91 & 0.49 & 16.15 & 93.14 \\
\bottomrule
\end{tabular*}
\end{table*}

\begin{figure}[!h]
    \centering
    \includegraphics[width=\columnwidth]{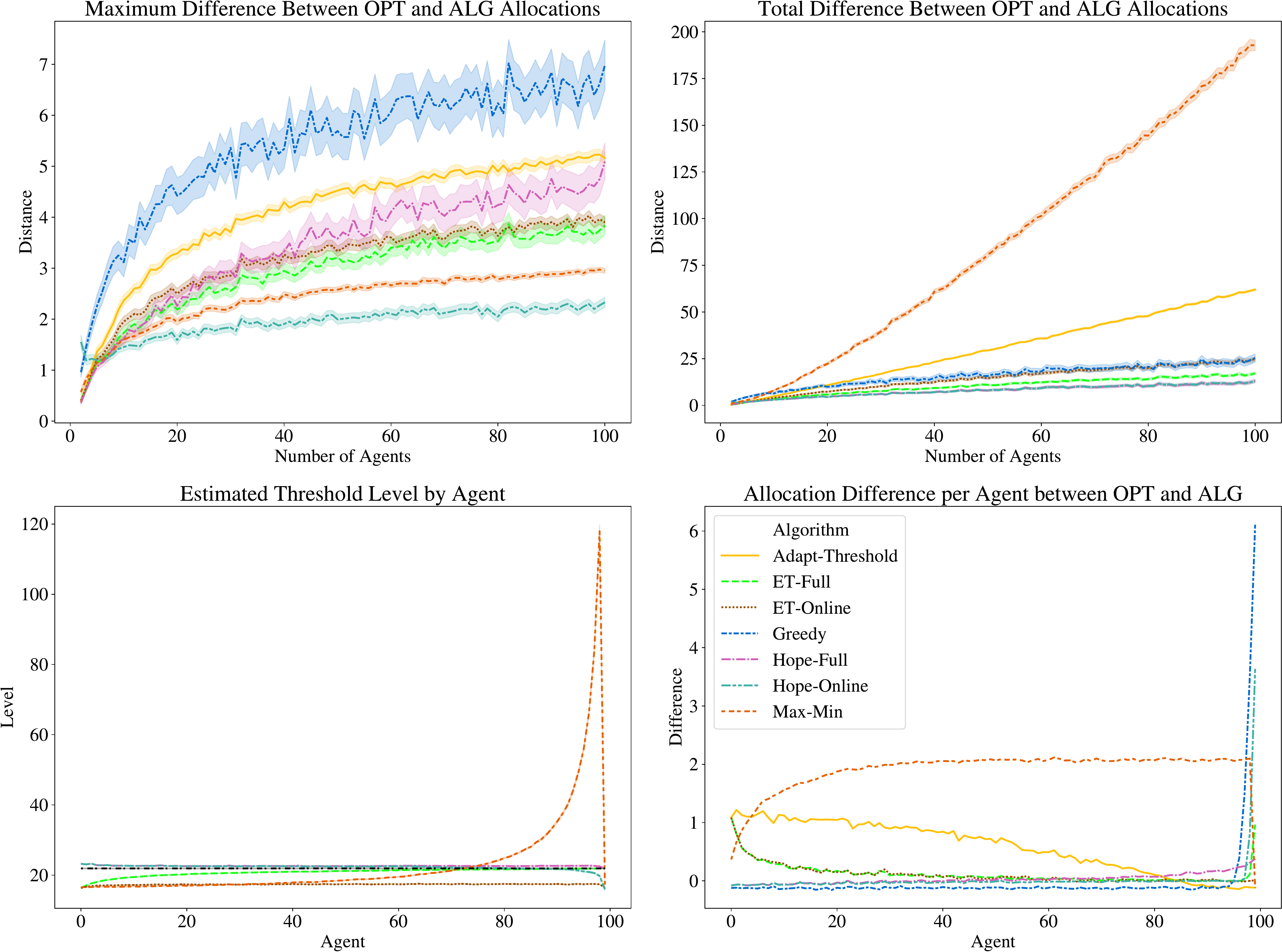}
    \caption{Comparison of \HopeOnline, \HopeFull, \EtOnline, \EtFull, \MaxMin, and the \textsc{Greedy} and \textsc{Adapt-Threshold} algorithms on a synthetic dataset where each type $\theta_i \sim \text{Gaussian}(15, 3)$ where the Gaussian distribution was discretized into twenty buckets.  Top left: comparison of $\Exp{\norm{X^{opt} - X^{alg}}_{\infty}}$ for the different algorithms as we scale $n$ from $1$ to $100$.  Top right: comparison of $\Exp{\norm{X^{opt} - X^{alg}}_{1}}$ for the different algorithms as we scale $n$ from $1$ to $100$.  Bottom left: comparison of the threshold used in the allocation for different groups, averaged over many simulations with $n = 100$ agents.  Bottom right: comparison of the agent by agent allocation difference, $\Exp{|X^{opt}_i - X^{alg}_i|}$ for the different agents with a fixed $n = 100$ agents.}
    \label{fig:gaussian}
\end{figure}

\begin{figure}[!h]
    \centering
    \includegraphics[width=\columnwidth]{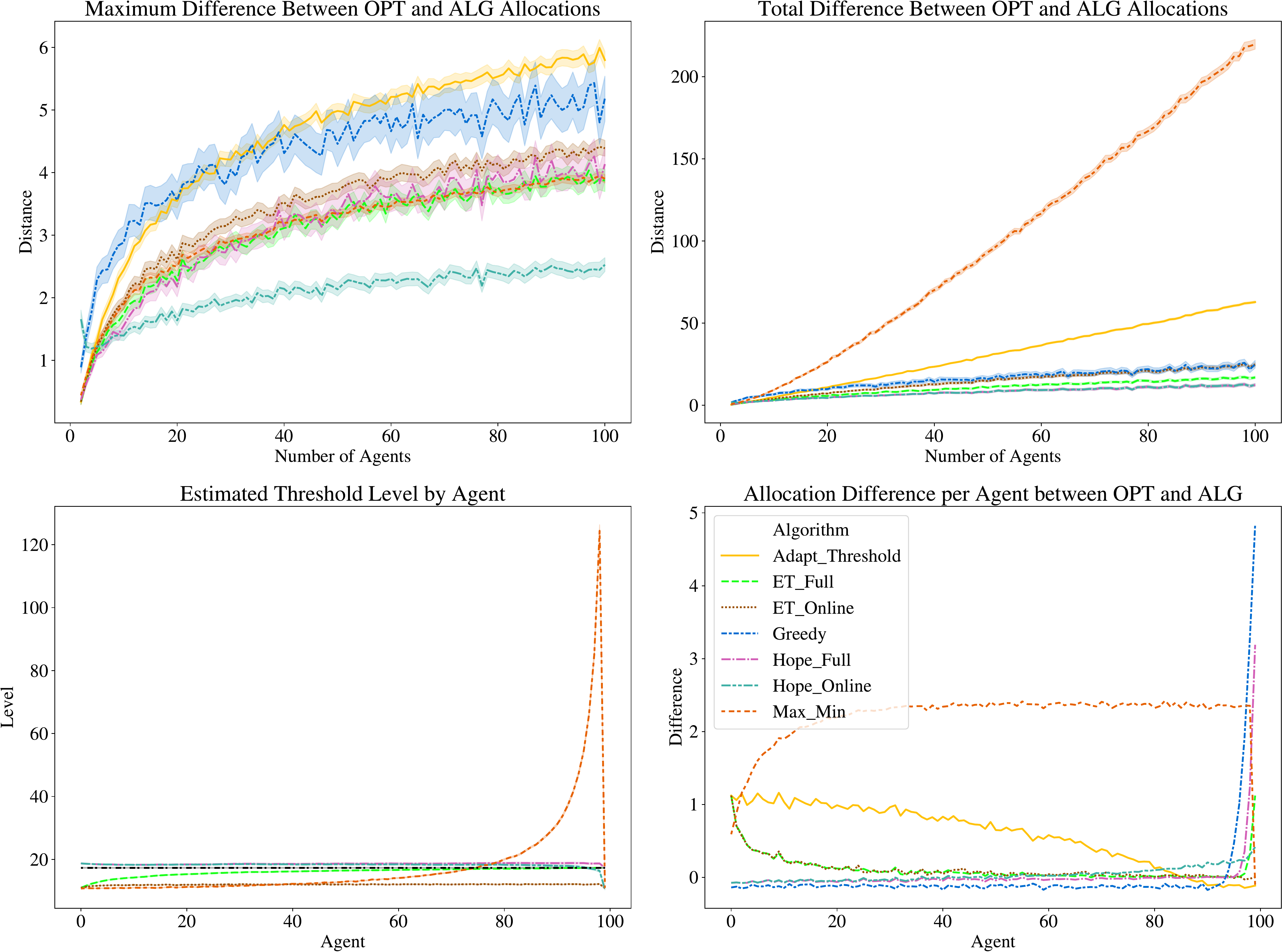}
    \caption{Comparison of \HopeOnline, \HopeFull, \EtOnline, \EtFull, \MaxMin, and the \textsc{Greedy} and \textsc{Adapt-Threshold} algorithms on a synthetic dataset where each type $\theta_i \sim \text{Poisson}(10)$ where the Poisson distribution was discretized into twenty buckets.  Top left: comparison of $\Exp{\norm{X^{opt} - X^{alg}}_{\infty}}$ for the different algorithms as we scale $n$ from $1$ to $100$.  Top right: comparison of $\Exp{\norm{X^{opt} - X^{alg}}_{1}}$ for the different algorithms as we scale $n$ from $1$ to $100$.  Bottom left: comparison of the threshold used in the allocation for different groups, averaged over many simulations with $n = 100$ agents.  Bottom right: comparison of the agent by agent allocation difference, $\Exp{|X^{opt}_i - X^{alg}_i|}$ for the different agents with a fixed $n = 100$ agents.}
    \label{fig:poisson}
\end{figure}

\begin{figure}[!h]
    \centering
    \includegraphics[width=\columnwidth]{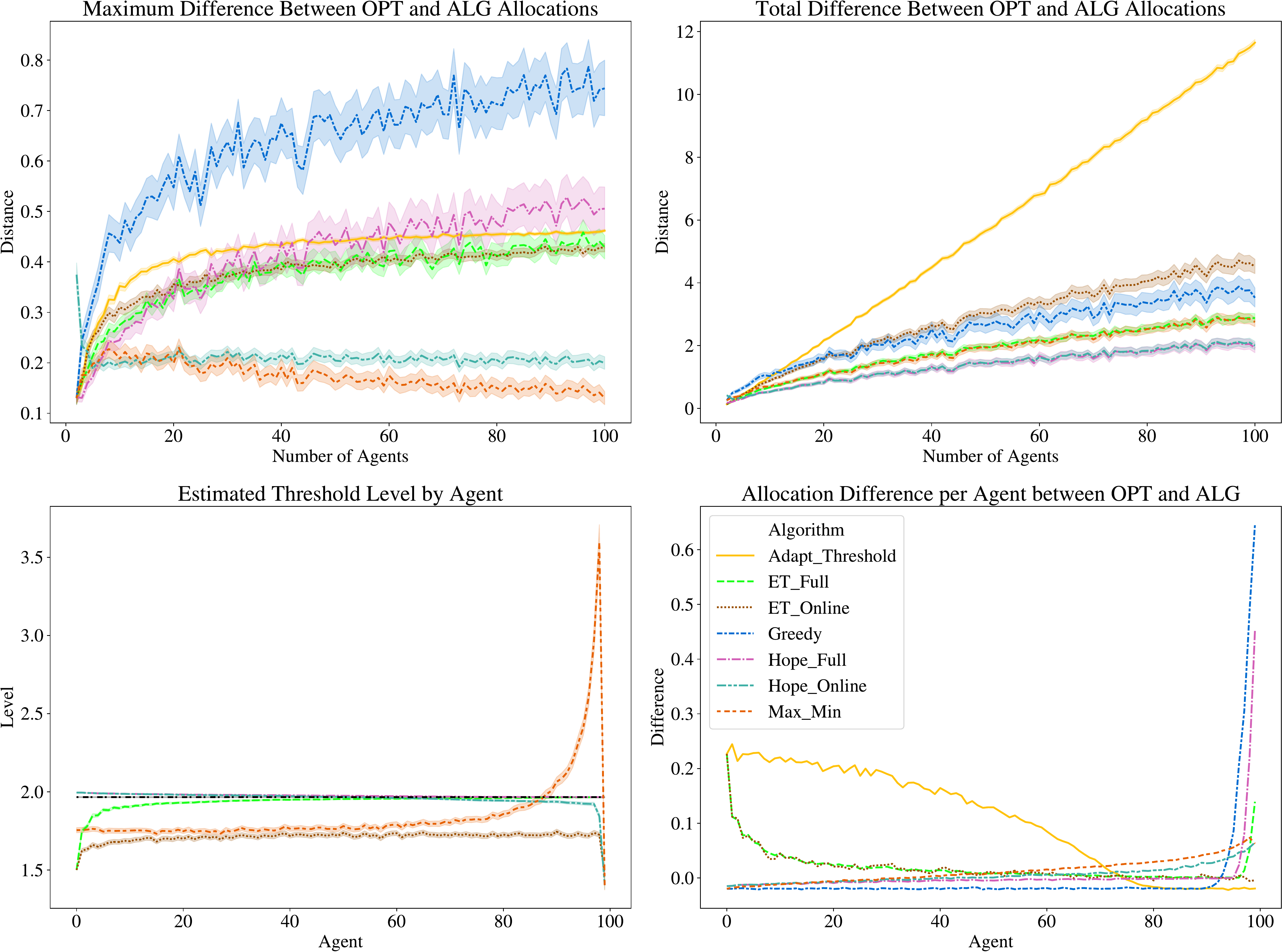}
    \caption{Comparison of \HopeOnline, \HopeFull, \EtOnline, \EtFull, \MaxMin, and the \textsc{Greedy} and \textsc{Adapt-Threshold} algorithms on a synthetic dataset where each type $\theta_i \sim \text{Uniform}\{1,2\}$.  Top left: comparison of $\Exp{\norm{X^{opt} - X^{alg}}_{\infty}}$ for the different algorithms as we scale $n$ from $1$ to $100$.  Top right: comparison of $\Exp{\norm{X^{opt} - X^{alg}}_{1}}$ for the different algorithms as we scale $n$ from $1$ to $100$.  Bottom left: comparison of the threshold used in the allocation for different groups, averaged over many simulations with $n = 100$ agents.  Bottom right: comparison of the agent by agent allocation difference, $\Exp{|X^{opt}_i - X^{alg}_i|}$ for the different agents with a fixed $n = 100$ agents.}
    \label{fig:simple}
\end{figure}
\end{document}